\documentclass[letterpaper,12pt]{article}
\usepackage{amssymb, amsmath}
\usepackage{amsthm}
\usepackage{algorithm}
\usepackage{algorithmic}
\usepackage{enumerate}
\usepackage[numbers,square,comma,sort&compress]{natbib}
\usepackage[pdftex,colorlinks]{hyperref}
\usepackage{a4}

\newtheorem{theorem}{Theorem}
\newtheorem{corollary}[theorem]{Corollary}
\newtheorem{fact}[theorem]{Fact}
\newtheorem{lemma}[theorem]{Lemma}

\newcommand{\comment}[1]{}

\begin{document}
\title{Metric $1$-median selection: Query complexity vs.\ approximation
ratio}

\author{
Ching-Lueh Chang \footnote{Department of Computer Science and
Engineering,
Yuan Ze University, Taoyuan, Taiwan. Email:
clchang@saturn.yzu.edu.tw}
\footnote{Innovation Center for Big Data and Digital Convergence,
Yuan Ze University, Taoyuan, Taiwan.}
}


\maketitle

\begin{abstract}
Consider the
problem
of finding
a point in a metric space $(\{1,2,\ldots,n\},d)$
with the minimum average distance to other points.
We show that this problem has no deterministic $o(n^{1+1/(h-1)})$-query
$(2h-\Omega(1))$-approximation algorithms for any constant
$h\in\mathbb{Z}^+\setminus\{1\}$.
\end{abstract}


\section{Introduction}

The {\sc metric $1$-median} problem asks for a point in
an $n$-point
metric space
with the minimum average distance to other points.
It
has a Monte-Carlo $O(n/\epsilon^2)$-time $(1+\epsilon)$-approximation
algorithm for all $\epsilon>0$~\cite{Ind99, Ind00}.
In
$\mathbb{R}^D$,
Kumar et al.~\cite{KSS10}
give
a Monte-Carlo $O(2^{\text{poly}(1/\epsilon)}D)$-time
$(1+\epsilon)$-approximation algorithm
for $1$-median selection and another algorithm for $k$-median
selection, where $D\ge 1$ and $\epsilon>0$.
Guha et al.~\cite{GMMMO03}
give streaming approximation algorithms for
$k$-median
selection
in metric spaces.

Chang~\cite{Cha13}, Wu~\cite{Wu14}
and Chang~\cite{Cha15median}
show that {\sc metric $1$-median} has a deterministic nonadaptive
$O(n^{1+1/h})$-time
$(2h)$-approximation algorithm for all constants $h\in\mathbb{Z}^+\setminus\{1\}$.
Furthermore,
Chang~\cite{Cha14arXiv}
shows the nonexistence of deterministic $o(n^2)$-time
$(4-\Omega(1))$-approximation
algorithms for {\sc metric $1$-median}.
This paper generalizes his result to show that {\sc metric $1$-median}
has no deterministic $o(n^{1+1/(h-1)})$-query $(2h-\Omega(1))$-approximation
algorithms for any constant $h\in\mathbb{Z}^+\setminus\{1\}$.
Combining
our result with
an existing
upper bound~\cite{Wu14, Cha15median},
{\small 
\begin{eqnarray}
&&
\min\left\{c\ge1\mid \text{{\sc metric $1$-median} has a deterministic
$O(n^{1+\epsilon})$-query $c$-approx.\ alg.}\right\}\nonumber\\
&=&\min\left\{c\ge1\mid \text{{\sc metric $1$-median} has a deterministic
$O(n^{1+\epsilon})$-time $c$-approx. alg.}\right\}\nonumber\\
&=&2\left\lceil\frac{1}{\epsilon}\right\rceil\nonumber
\end{eqnarray}
}
for all constants $\epsilon\in(0,1)$.
That is, we determine the best approximation ratio of
deterministic $O(n^{1+\epsilon})$-query (resp., $O(n^{1+\epsilon})$-time)
algorithms for {\em all} $\epsilon\in(0,1)$.

As in the previous lower bounds
for deterministic algorithms~\cite{Cha14arXiv, Cha12},
we use
an
adversarial method.
Roughly speaking, our proof proceeds as follows:
\begin{enumerate}[(i)]
\item Design an adversary {\sf Adv} for answering
the distance queries of any deterministic
algorithm $A$ with query complexity $q(n)=o(n^{1+1/(h-1)})$.
\item\label{itemlargedistanceforalgorithmoutput} Show that $A$'s output has a large average distance to other points,
according to
{\sf Adv}'s answers to $A$.
\item\label{qualityoftheoptimalpointitem}
Construct a distance function with respect to which a certain point
$\hat{\alpha}$
has a small average distance to other points.
\item\label{itemconstructionofthefinalmetric}
Construct the final distance function $d(\cdot,\cdot)$
similar to that in item~(\ref{qualityoftheoptimalpointitem}).
\item Show that $d$ is a metric.
\item\label{itemconsistencyofdistancemetricwiththeanswersoftheadversary}
Show the consistency of $d(\cdot,\cdot)$
with
{\sf Adv}'s
answers.
\item\label{itemputtingthingstogether}
Compare
$\hat{\alpha}$
in item~(\ref{qualityoftheoptimalpointitem})
with $A$'s output
to establish our lower bound
on $A$'s approximation ratio.
\end{enumerate}

Central to our constructions
are two
graph sequences, $\{H^{(i)}\}_{i=0}^{q(n)}$
and $\{G^{(i)}\}_{i=0}^{q(n)}$ in Sec.~\ref{mainresultsection},
that
are unseen in previous lower bounds~\cite{MP04, Cha12, Cha14arXiv}.
Like
in
\cite{Cha14arXiv},
we need a small set $S$ of points whose distances to other points
are answered as large values during $A$'s execution, and yet we
assign a small value to the distances from a certain point $\hat{\alpha}\in S$
to many other points
in item~(\ref{qualityoftheoptimalpointitem}).

This paper is organized as follows.
Sec.~\ref{definitionssection}
introduces the terminologies.
Sec.~\ref{mainresultsection} proves our main theorem that {\sc metric $1$-median}
has no deterministic $o(n^{1+1/(h-1)})$-query $(2h-\Omega(1))$-approximation
algorithms for any constant $h\in\mathbb{Z}^+\setminus\{1\}$.
In particular,
Secs.~\ref{outputofthealgorithmisbadsubsection},~\ref{agoodpointsubsection},~\ref{consistencywithametricsubsection}~and~\ref{puttingthingstogethersubsection}
correspond to
items~(\ref{itemlargedistanceforalgorithmoutput}),~(\ref{qualityoftheoptimalpointitem}),~(\ref{itemconstructionofthefinalmetric})--(\ref{itemconsistencyofdistancemetricwiththeanswersoftheadversary})~and~(\ref{itemputtingthingstogether})
above,
respectively.

\comment{ 
it first presents an adversary for answering
the
distance queries of any deterministic
$o(n^{1+1/(h-1)})$-query algorithm $A$.
Then Sec.~\ref{outputofthealgorithmisbadsubsection}
shows that $A$'s output has a large average distance to other points,
where all distances are as answered by the adversary.
Sec.~\ref{agoodpointsubsection} constructs a distance function with
respect to which a certain point
$\hat{\alpha}$
has a small average distance to other points.
Roughly speaking,
Sec.~\ref{consistencywithametricsubsection}
constructs the final metric $d$ with respect to which
the conclusions of
Secs.~\ref{outputofthealgorithmisbadsubsection}--\ref{agoodpointsubsection}
remain true.
Finally, Sec.~\ref{puttingthingstogethersubsection}
puts everything together for our results.
}

\section{Definitions}\label{definitionssection}

A
finite
metric space $(M,d)$ is a
finite
set $M$ endowed with a function
$d\colon M^2\to[0,\infty)$
such that
\begin{itemize}
\item $d(x,x)=0$,
\item $d(x,y)>0$ if $x\neq y$,
\item $d(x,y)=d(y,x)$, and
\item $d(x,y)+d(y,z)\ge d(x,z)$
\end{itemize}
for all $x$, $y$, $z\in M$~\cite{Rud76}.
For
all
$c\ge1$,
a point $z\in M$ is
said to be
a $c$-approximate $1$-median of $(M,d)$
if
$$\sum_{x\in M}\,d\left(z,x\right)\le c\cdot\sum_{x\in M}\,d\left(y,x\right)$$
for all $y\in M$.
For convenience,
$[n]\stackrel{\text{def.}}{=}\{1,2,\ldots,n\}$.

For
deterministic
algorithms
$A$ and
${\cal O}\colon\{1,2,\ldots,n\}^2\to\mathbb{R}$,
denote by
$A^{\cal O}(1^n)$
the execution of $A$ with oracle
access to
$\cal O$
and with input $1^n$, where $n\in\mathbb{N}$.
As the input to $A$ will be $1^n$ throughout this paper,
abbreviate
$A^{\cal O}(1^n)$ as $A^{\cal O}$.
If
$A^d$
outputs a $c$-approximate $1$-median
of $([n],d)$
for each finite metric space $([n],d)$,
then
$A$ is said to be $c$-approximate for {\sc metric $1$-median},
where $c\ge1$.

\begin{fact}[\cite{Cha13, Cha15median, Wu14}]
\label{existingupperbound}
For each constant $h\in\mathbb{Z}^+\setminus\{1\}$,
{\sc metric $1$-median}
has a deterministic nonadaptive $O(n^{1+1/h})$-time
$(2h)$-approximation algorithm.
\end{fact}

\comment{ 
The following facts are well-known.

\begin{fact}
For
any
positively-weighted
undirected graph $G=(V,E,w)$ and $c>0$,
the function $d\colon V\times V\to [\,0,\infty\,)$
defined by
$(x,y)\mapsto \min\{d_G(x,y),c\}$
is a metric.
\end{fact}
}

A weighted undirected graph $G=(V,E,w)$
has a finite vertex set $V$, an edge set $E$ and a weight function
$w\colon E\to(0,\infty)$, where each edge is an unordered pair
of distinct vertices in $V$.
If $w\colon Y\to(0,\infty)$ for a superset $Y$ of $E$,
interpret
$(V,E,w)$
simply
as $(V,E,w|_E)$,
where
$w|_E$ denotes the restriction of $w$ on $E$.
For all $v\in V$,
let
$$N_G(v)\stackrel{\text{def.}}{=}\left\{u\in V\mid (u,v)\in E\right\}$$
and $\text{\rm deg}_G(v)\stackrel{\text{def.}}{=}|N_G(v)|$.
For all $S\subseteq V$,
$N_G(S)\stackrel{\text{def.}}{=}\bigcup_{v\in S}\,N_G(v)$.
For all $s$, $t\in V$,
an $s$-$t$ path $P$ in $G$
is a sequence
$\{v_i\in V\}_{i=0}^k$ satisfying
$k\in\mathbb{N}$,
$v_0=s$,
$v_k=t$
and
$(v_i,v_{i+1})\in E$ for all $i\in\{0,1,\ldots,k-1\}$.
Its weight (or length)
is $w(P)\stackrel{\text{def.}}{=}\sum_{i=0}^{k-1}\,w(v_i,v_{i+1})$.\footnote{$w(P)$
is
a common and convenient
abuse
of notation.}
The shortest $s$-$t$ distance in $G$
is
$$d_{G}(s,t)=\inf\left\{w(P)\mid
\text{$P$ is an $s$-$t$ path in $G$}\right\},$$
where $s$, $t\in V$.
So $d_G(s,t)=\infty$ if $G$ has no $s$-$t$ paths.
Note that we allow only positive weights, i.e.,
$\mathop{\mathrm{Im}}(w)\subseteq (0,\infty)$.
So a shortest $s$-$t$ path must be simple,
i.e., it does not repeat vertices.
If $w\equiv 1$,
abbreviate $(V,E,w)$ as $(V,E)$ and
call it an unweighted graph.

The following fact is well-known.

\begin{fact}\label{basicgraphtheoryfact}
For each undirected graph $G=(V,E)$,
$$\sum_{v\in V}\,\text{\rm deg}_G(v)=2\cdot|E|.$$
\end{fact}

For a predicate $P$, let $\chi[P]=1$ if $P$ is true and
$\chi[P]=0$ otherwise.
The following fact
about geometric series
is not hard to see.

\begin{fact}\label{geometricseriesbound}
For all $r\ge 2$ and $m\in\mathbb{N}$,
$$\sum_{k=0}^m\,r^k\le 2r^m.$$
\end{fact}

\section{Query complexity vs.\ approximation ratio}\label{mainresultsection}

Throughout this section,
\begin{itemize}
\item
$n\in\mathbb{Z}^+$,
\item $\delta\in(0,1)$ and $h\in\mathbb{Z}^+\setminus\{1\}$ are constants (i.e.,
they are independent of $n$),
\item $A$ is a deterministic
$o(n^{1+1/(h-1)})$-query
algorithm for {\sc metric $1$-median},
and
\item
$S=[\lfloor\delta n\rfloor]\subseteq [n]$.
\end{itemize}

All
pairs in $[n]^2$
are assumed
to be
unordered in this section.
So,
e.g.,
$(1,2)\in\{2\}\times[n]$.
By padding
at most $n-1$
dummy
queries, assume
without loss of generality
that $A$
will have queried
for the
distances between its output and all other points
when
halting.
Denote
$A$'s query complexity by
$$q(n)=o\left(n^{1+1/(h-1)}\right).$$
Without loss of generality,
forbid
making the same query twice or querying
for the distance
from a point to itself,
where the queries for $d(x,y)$ and $d(y,x)$
are considered to be the same for $x$, $y\in [n]$.
Furthermore,
let $n$ be sufficiently large
to satisfy
\begin{eqnarray}
q(n)&\le& \delta n^{1+1/(h-1)},\label{tediouscondition1}\\
\delta n^{1/(h-1)}&>&3,\label{tediouscondition2}\\
\frac{2 q(n)}{|S|-1}&\le&\delta n^{1/(h-1)}.\label{tediouscondition3}
\end{eqnarray}


\comment{ 
\begin{eqnarray}
w\left(u,v\right)
\stackrel{\text{def.}}{=}
\left\{
\begin{array}{ll}
2, & \text{if $u\notin S$ and $v\notin S$,}\\
2h-1, & \text{otherwise}.
\end{array}
\right.
\label{weightfunction}
\end{eqnarray}
}
Define two unweighted
undirected graphs $G^{(0)}$
and $H^{(0)}$ by
\begin{eqnarray}
E_G^{(0)}
&\stackrel{\text{def.}}{=}&
\left\{\left(u,v\right)
\mid \left(u,v\in [n]\setminus S\right)
\land\left(u\neq v\right)\right\},
\label{completegraphedgeset}\\
G^{(0)}
&\stackrel{\text{def.}}{=}&\left([n],E_G^{(0)}\right),
\label{initiallargegraph}\\
E_H^{(0)}
&\stackrel{\text{def.}}{=}&
\emptyset,
\label{initiallymarkededgeset}\\
H^{(0)}
&\stackrel{\text{def.}}{=}&
\left([n],E_H^{(0)}\right).\label{initiallymarkedgraph}
\end{eqnarray}

\comment{ 
\begin{lemma}
For all $(u_1,u_2,\ldots,u_h)$, $(v_1,v_2,\ldots,v_h)\in T^h$,
$$
d_{H^{(0)}}\left(
\left(u_1,u_2,\ldots,u_h\right), \left(v_1,v_2,\ldots,v_h\right)
\right)\le 2h.
$$
\end{lemma}
\begin{proof}
By equations~(\ref{initiallymarkededgeset})--(\ref{initiallymarkedgraph}),
$H^{(0)}$ has the path
$P$ whose $i$th vertex is
$$\left(u_1,u_2,\ldots,u_{h+1-i},v_{h-i+2},v_{h-i+3},\ldots,v_h\right)$$
for $i\in[h+1]$.
Clearly, $P$ has exactly $h$ edges, each of weight $2$ w.r.t.\ $w$.
\end{proof}
}

\begin{figure}
\begin{algorithmic}[1]
\STATE Let $E_G^{(0)}$, $G^{(0)}$, $E_H^{(0)}$ and $H^{(0)}$
be as in equations~(\ref{completegraphedgeset})--(\ref{initiallymarkedgraph});
\FOR{$i=1$, $2$, $\ldots$, $q(n)$}
  \STATE Receive the $i$th query of $A$, denoted $(a_i,b_i)$;
  \IF{$d_{G^{(i-1)}}(a_i,b_i)\le h$}
    \STATE Find a shortest
    $a_i$-$b_i$ path $P_i$ in $G^{(i-1)}$;
    \STATE $E_H^{(i)}\leftarrow E_H^{(i-1)}\cup \{e\mid \text{$e$ is an edge on $P_i$}\}$;
    \STATE $H^{(i)}\leftarrow ([n],E_H^{(i)})$;
    \STATE $E_G^{(i)}\leftarrow E_G^{(i-1)}
\setminus\{(u,v)\in E_G^{(i-1)}\setminus E_H^{(i)}\mid
(\text{deg}_{H^{(i)}}(u)\ge \delta n^{1/(h-1)}-2)
\lor(\text{deg}_{H^{(i)}}(v)\ge \delta n^{1/(h-1)}-2)
\}
$;
    \STATE $G^{(i)}\leftarrow ([n],E_G^{(i)})$;
  \ELSE
    \STATE $E_H^{(i)}\leftarrow E_H^{(i-1)}$;
    \STATE $H^{(i)}\leftarrow ([n],E_H^{(i)})$;
    \STATE $E_G^{(i)}\leftarrow E_G^{(i-1)}$;
    \STATE $G^{(i)}\leftarrow ([n],E_G^{(i)})$;
\ENDIF
  \STATE $Q^{(i)}\leftarrow ([n],\{(a_j,b_j)\mid
j\in[i]\})$;
  \STATE Output
$\min\{d_{H^{(i)}}(a_i,b_i),
h-(1/2)\cdot\chi[\exists v\in\{a_i,b_i\},\,
(v\in S)\land (\text{deg}_{Q^{(i)}}(v)\le
\delta n^{1/(h-1)})]\}$
as the answer to the $i$th query
of $A$;
\ENDFOR
\end{algorithmic}
\caption{Algorithm {\sf Adv} for answering $A$'s queries}
\label{answeringalgorithm}
\end{figure}

Algorithm
{\sf Adv} in Fig.~\ref{answeringalgorithm}
answers
$A$'s
queries.
In particular,
for
all
$i\in[q(n)]$,
the
$i$th iteration of the loop of {\sf Adv}
answers the $i$th query of $A$, denoted $(a_i,b_i)\in [n]^2$.
It
constructs three unweighted undirected graphs,
$G^{(i)}=([n],E_G^{(i)})$,
$H^{(i)}=([n],E_H^{(i)})$ and $Q^{(i)}$.
As $G^{(i-1)}$
is unweighted for all $i\in[q(n)]$,
$P_i$ in line~5 of {\sf Adv} is an $a_i$-$b_i$ path
in $G^{(i-1)}$ with the minimum number of edges.
By line~16 of {\sf Adv}, the edges of $Q^{(i)}$
are precisely the first $i$ queries of $A$.

\begin{lemma}\label{chainofgraphs}
$$
E_H^{(0)}\subseteq E_H^{(1)}\subseteq \ldots\subseteq E_H^{(q(n))}
\subseteq E_G^{(q(n))}\subseteq E_G^{(q(n)-1)}\subseteq
\ldots\subseteq
E_G^{(0)}.
$$
\comment{
$$
H^{(0)}\subseteq H^{(1)}\subseteq \ldots\subseteq H^{(q(n))}
\subseteq G^{q(n)}\subseteq G^{q(n)-1}\subseteq G^{(0)}.
$$
}
\end{lemma}
\begin{proof}
By
lines~6~and~11
of {\sf Adv} in Fig.~\ref{answeringalgorithm},
$E_H^{(i-1)}\subseteq E_H^{(i)}$
for
all
$i\in[q(n)]$.
By
lines~8~and~13,
$E_G^{(i)}\subseteq E_G^{(i-1)}$
for all $i\in[q(n)]$.

To show that
$E_H^{(q(n))}\subseteq E_G^{(q(n))}$,
we shall prove
the stronger statement that
$E_H^{(i)}\subseteq E_G^{(i)}$
for all
$i\in\{0,1,\ldots,q(n)\}$
by mathematical induction.
By
equation~(\ref{initiallymarkededgeset}),
$E_H^{(0)}\subseteq E_G^{(0)}$.
Assume
as the induction hypothesis that
$E_H^{(i-1)}\subseteq E_G^{(i-1)}$.
The following
shows that $E_H^{(i)}\subseteq E_G^{(i-1)}$
by examining each $e\in E_H^{(i)}$:
\begin{enumerate}[{Case}~1:]
\item $e\in E_H^{(i-1)}$.
By the induction hypothesis,
$e\in E_G^{(i-1)}$.
\item $e\notin E_H^{(i-1)}$.
As $e\in E_H^{(i)}\setminus E_H^{(i-1)}$,
lines~6~and~11 show that
$e$ is on $P_i$
(and that the $i$th iteration
of the loop of {\sf Adv} runs
line~6 rather than line~11).
By line~5, each edge on $P_i$ is in $E_G^{(i-1)}$.
In particular,
$e\in E_G^{(i-1)}$.
\end{enumerate}
Having shown
that
$E_H^{(i)}\subseteq E_G^{(i-1)}$,
lines~8~and~13
will
both
result in
$E_H^{(i)}\subseteq E_G^{(i)}$,
completing the induction step.
\end{proof}

\comment{ 
The following lemma
shows that line~17 of {\sf Adv}
answers $A$'s queries consistently with
$\min\{d_{G^{(q(n))}}(\cdot,\cdot),2h\}$ and
$\min\{d_{H^{(q(n))}}(\cdot,\cdot),2h\}$.
}

\begin{lemma}\label{consistentanswers}
For all
$i\in[q(n)]$ with $d_{G^{(i-1)}}(a_i,b_i)\le h$,
$$
d_{H^{(i)}}\left(a_i,b_i\right)
=d_{H^{(q(n))}}\left(a_i,b_i\right)
=d_{G^{(q(n))}}\left(a_i,b_i\right)
=d_{G^{(i-1)}}\left(a_i,b_i\right).$$
\end{lemma}
\begin{proof}
By line~4 of {\sf Adv}, the $i$th iteration of the loop runs
lines~5--9.
Lines~5--7
put (the edges of)
a shortest $a_i$-$b_i$ path in $G^{(i-1)}$
into $H^{(i)}$; hence
$$d_{H^{(i)}}\left(a_i,b_i\right)\le d_{G^{(i-1)}}\left(a_i,b_i\right).$$
This and Lemma~\ref{chainofgraphs} complete the proof.
\end{proof}

Below is an easy consequence
of Lemma~\ref{chainofgraphs}.

\begin{lemma}\label{ifthesmallestdistancerunsoutthenothersdoso}
For all $i\in[q(n)]$ with $d_{G^{(i-1)}}(a_i,b_i)>h$,
$$d_{G^{(q(n))}}(a_i,b_i)>h.$$
\end{lemma}

\subsection{The average distance from $A$'s output to other points}
\label{outputofthealgorithmisbadsubsection}

This subsection shows that the output of
$A^{\sf Adv}$
has a large average distance
to other points,
according to the answers of
{\sf Adv}.

\begin{lemma}\label{increaseinnumberofincidentmarkededges}
For all
$i\in[q(n)]$
and $v\in [n]$,
$$\text{\rm deg}_{H^{(i)}}(v)
\le\text{\rm deg}_{H^{(i-1)}}(v)+2.$$
\end{lemma}
\begin{proof}
If the $i$th iteration of the loop of {\sf Adv} runs lines~11--14
but not 5--9,
then $H^{(i)}=H^{(i-1)}$, proving the lemma.
So assume otherwise.
Being shortest,
$P_i$
in line~5
does not repeat vertices.
Therefore, $v$ is incident to at most two edges on
$P_i$, which together with lines~6--7
complete the proof.
\end{proof}

\begin{lemma}\label{boundonnumberofincidentedges}
For all $v\in [n]$,
$$\text{\rm deg}_{H^{(q(n))}}(v)< \delta n^{1/(h-1)}.$$
\end{lemma}
\begin{proof}
Assume
\begin{eqnarray}
\text{\rm deg}_{H^{(q(n))}}(v)\ge \delta n^{1/(h-1)}-2
\label{justatrivialconditionfortheretobesomthingtoprove}
\end{eqnarray}
for,
otherwise, there is nothing to prove.
Clearly,
\begin{eqnarray}
\text{\rm deg}_{H^{(0)}}(v)
\stackrel{\text{(\ref{initiallymarkededgeset})--(\ref{initiallymarkedgraph})}}{=}0
\stackrel{\text{(\ref{tediouscondition2})}}{<}\delta n^{1/(h-1)}-2.
\label{ofcoursetheoriginaldegreeislow}
\end{eqnarray}
By
inequalities~(\ref{justatrivialconditionfortheretobesomthingtoprove})--(\ref{ofcoursetheoriginaldegreeislow}),
there exists
$i\in[q(n)]$
satisfying
\begin{eqnarray}
\text{\rm deg}_{H^{(i-1)}}(v)< \delta n^{1/(h-1)}-2,
\label{beforethejump}\\
\text{\rm deg}_{H^{(i)}}(v)\ge \delta n^{1/(h-1)}-2.
\label{afterthejump}
\end{eqnarray}

Clearly,
\begin{eqnarray}
N_{G^{(i)}}(v)
=\left\{u\in[n]
\mid\left(u,v\right)\in E_G^{(i)}\right\}.
\label{easyfromthedefinitionofG}
\end{eqnarray}
As $H^{(i-1)}\neq H^{(i)}$ by
inequalities~(\ref{beforethejump})--(\ref{afterthejump}),
the $i$th iteration of the loop of {\sf Adv} runs
lines~5--9 but not 11--14.
By
inequality~(\ref{afterthejump}) and
line~8
of {\sf Adv},
\begin{eqnarray}
\left\{u\in[n]
\mid\left(u,v\right)\in E_G^{(i)}\right\}
=
\left\{u\in [n]
\mid \left(u,v\right)\in E_G^{(i-1)}\setminus\left(E_G^{(i-1)}
\setminus E_H^{(i)}\right)
\right\}.
\label{setofneighborsinGi}
\end{eqnarray}
Equations~(\ref{easyfromthedefinitionofG})--(\ref{setofneighborsinGi})
and Lemma~\ref{chainofgraphs} give
\begin{eqnarray}
N_{G^{(i)}}(v)
=
\left\{u\in [n]
\mid \left(u,v\right)\in E_H^{(i)}
\right\}.\label{setofneighbors}
\end{eqnarray}

By
inequality~(\ref{beforethejump}) and
Lemma~\ref{increaseinnumberofincidentmarkededges},
\begin{eqnarray}
\text{\rm deg}_{H^{(i)}}(v)< \delta n^{1/(h-1)}.
\nonumber
\end{eqnarray}
This and equation~(\ref{setofneighbors})
imply
$\text{\rm deg}_{G^{(i)}}(v)< \delta n^{1/(h-1)}$, which together with
Lemma~\ref{chainofgraphs} completes the proof.
\end{proof}

\begin{lemma}\label{averagedistanceishighinthemarkedgraph}
For all $v\in [n]$,
$$
\left|\left\{
u\in [n]\mid d_{H^{(q(n))}}\left(v,u\right)<h
\right\}\right|
\le2\delta^{h-1}n.
$$
\end{lemma}
\begin{proof}
By Lemma~\ref{boundonnumberofincidentedges},
$$
\left|
\left\{
u\in [n]\mid \text{$\exists$ $v$-$u$ path in $H^{(q(n))}$ with exactly $k$ edges}
\right\}
\right|
\le \left(\delta n^{1/(h-1)}\right)^k
$$
for all $k\in\mathbb{N}$.
Consequently,
{\small 
\begin{eqnarray}
\left|
\left\{
u\in [n]\mid \text{$\exists$ $v$-$u$ path in $H^{(q(n))}$ with at most $h-1$ edges}
\right\}
\right|
&\le& \sum_{k=0}^{h-1}\,\left(\delta n^{1/(h-1)}\right)^k
\nonumber\\
&\stackrel{\text{(\ref{tediouscondition2})~and~Fact~\ref{geometricseriesbound}}}{\le}&
2\delta^{h-1}n.
\nonumber
\end{eqnarray}
}
Finally, recall that $H^{(q(n))}$ is unweighted.
\end{proof}

Denote
the output of
$A^{\text{\sf Adv}}$
by $z$.
Furthermore,
\begin{eqnarray}
I\stackrel{\text{def.}}{=}\left\{
j\in \left[q(n)\right]
\mid z\in\left\{a_j,b_j\right\}
\right\}.
\label{indexset}
\end{eqnarray}
The following lemma analyzes the sum of
the distances,
as answered by line~17 of {\sf Adv},
from $z$
to other points.

\begin{lemma}\label{theclosenesscentralityoftheoutputofthealgorithm}
{\small 
\begin{eqnarray}
&&
\sum_{i\in I}\,
\min\left\{d_{H^{(i)}}\left(a_i,b_i\right),
h-\frac{1}{2}\cdot
\chi\left[\exists v\in
\left\{a_i,b_i\right\},
\,
\left(v\in S\right)\land
\left(\text{\rm deg}_{Q^{(i)}}(v)
\le \delta n^{1/(h-1)}\right)\right]\right\}\nonumber\\
&\ge& n\cdot \left(h-2h\delta^{h-1}-o(1)
-\delta\right).
\nonumber
\end{eqnarray}
}
\comment{ 
\begin{eqnarray}
&&
\sum_{j=1}^{n-1}\,
\min\left\{d_{G^{(q(n))}}\left(a_{f(j)},b_{f(j)}\right),
h-\frac{1}{2}\cdot
\chi\left[\exists v\in
\left\{a_{f(j)},b_{f(j)}\right\},
\,
\left(v\in S\right)\land
\left(\text{\rm deg}_{G^{(f(u))}}(v)
\le \delta n^{1/(h-1)}\right)\right]\right\}\nonumber\\
&\ge& h\cdot \left(1-2\delta^{h-1}\right).
\nonumber
\end{eqnarray}
}
\comment{ 
$$
\sum_{u\in [n]}\,
\min\left\{d_{G^{(q(n))}}\left(u,z\right),
h-\frac{1}{2}\cdot
\chi\left[\exists v\in
\left\{u,z\right\},
\,
\left(v\in S\right)\land
\left(\text{\rm deg}_{G^{(f(u))}}(v)
\le \delta n^{1/(h-1)}\right)\right]\right\}
\ge h\cdot \left(1-2\delta^{h-1}\right).
$$
}
\end{lemma}
\begin{proof}
By
Lemma~\ref{chainofgraphs},
{\small 
\begin{eqnarray}
&&
\sum_{i\in I}\,
\min\left\{d_{H^{(i)}}\left(a_i,b_i\right),
h-\frac{1}{2}\cdot
\chi\left[\exists v\in
\left\{a_i,b_i\right\},
\,
\left(v\in S\right)\land
\left(\text{\rm deg}_{Q^{(i)}}(v)
\le \delta n^{1/(h-1)}\right)\right]\right\}
\label{firstpiece}\\
&\ge&
\sum_{i\in I}\,
\min\left\{d_{H^{(q(n))}}\left(a_i,b_i\right),
h-\frac{1}{2}\cdot
\chi\left[\exists v\in
\left\{a_i,b_i\right\},
\,
\left(v\in S\right)\land
\left(\text{\rm deg}_{Q^{(i)}}(v)
\le \delta n^{1/(h-1)}\right)\right]\right\}
\,\,\,\,\,\,\,\,\,\,\nonumber\\
&\ge& \sum_{i\in I}\,
\min\left\{d_{H^{(q(n))}}\left(a_i,b_i\right),
h
\right\}\nonumber\\
&-&\sum_{i\in I}\,\frac{1}{2}\cdot
\chi\left[\exists v\in
\left\{a_i,b_i\right\},
\,
\left(v\in S\right)\land
\left(\text{\rm deg}_{Q^{(i)}}(v)
\le \delta n^{1/(h-1)}\right)\right].
\nonumber
\end{eqnarray}
} 


For all $i\in I$, there exists $c_i\in[n]$ with $\{z,c_i\}
=\{a_i,b_i\}$ by equation~(\ref{indexset}).
Therefore,
$$\sum_{i\in I}\,
\min\left\{
d_{H^{(q(n))}}\left(a_i,b_i\right), h\right\}
=\sum_{i\in I}\,
\min\left\{
d_{H^{(q(n))}}\left(z,c_i\right), h\right\}.
$$
As we forbid repeated queries,
$\{c_i\}_{i\in I}$ is a sequence of distinct points.
So by
Lemma~\ref{averagedistanceishighinthemarkedgraph},
$$\sum_{i\in I}\,
\min\left\{
d_{H^{(q(n))}}\left(z,c_i\right), h\right\}
\ge
h\cdot \left(|I| - 2\delta^{h-1}n\right).
$$
Recall that
$A^{\sf Adv}$
will have queried for the distances
between its output (which is $z$)
and all other points when halting.
So
$$|I|\ge n-1$$
by equation~(\ref{indexset}).\footnote{Because we forbid repeated
queries and queries for the distance from a point to itself,
we also have $|I|\le n-1$.}

Clearly,
\begin{eqnarray}
&&\sum_{i\in I}\,
\chi\left[\exists v\in
\left\{a_i,b_i\right\},
\,
\left(v\in S\right)\land
\left(\text{\rm deg}_{Q^{(i)}}(v)
\le \delta n^{1/(h-1)}\right)\right]
\nonumber\\
&=&
\sum_{i\in I}\,
\chi\left[\exists v\in
\left\{z,c_i\right\},
\,
\left(v\in S\right)\land
\left(\text{\rm deg}_{Q^{(i)}}(v)
\le \delta n^{1/(h-1)}\right)\right]
\nonumber\\
&\le&
\sum_{i\in I}\,
\chi\left[\left(z\in S\right)\land
\left(\text{\rm deg}_{Q^{(i)}}\left(z\right)
\le \delta n^{1/(h-1)}\right)\right]\nonumber\\
&+&\sum_{i\in I}\,
\chi\left[\left(c_i\in S\right)\land
\left(\text{\rm deg}_{Q^{(i)}}\left(c_i\right)
\le \delta n^{1/(h-1)}\right)\right].
\nonumber
\end{eqnarray}

By
line~16 of {\sf Adv} and
equation~(\ref{indexset}),
\begin{eqnarray}
\text{deg}_{Q^{(i)}}\left(z\right)
=\left|\left\{j\in I\mid j\le i\right\}\right|.
\nonumber
\end{eqnarray}
Therefore,
\begin{eqnarray}
\sum_{i\in I}\,
\chi\left[\left(z\in S\right)\land
\left(\text{\rm deg}_{Q^{(i)}}\left(z\right)
\le \delta n^{1/(h-1)}\right)\right]
&\le&
\sum_{i\in I}\,
\chi\left[\left|\left\{j\in I\mid j\le i\right\}\right|
\le \delta n^{1/(h-1)}\right]
\nonumber\\
&\le&\delta n^{1/(h-1)},
\nonumber
\end{eqnarray}
where the last inequality follows because
$|\{j\in I\mid j\le i\}|=k$
when $i$ is
the $k$th smallest element of $I$,
for all $k\in[|I|]$.
Recall the distinctness of the points in $\{c_i\}_{i\in I}$.
Therefore,
\begin{eqnarray}
\sum_{i\in I}\,
\chi\left[\left(c_i\in S\right)\land
\left(\text{\rm deg}_{Q^{(i)}}\left(c_i\right)
\le \delta n^{1/(h-1)}\right)\right]
\le
\sum_{i\in I}\,
\chi\left[c_i\in S\right]
\le|S|=\left\lfloor\delta n\right\rfloor.
\label{lastpiece}
\end{eqnarray}
Inequalities~(\ref{firstpiece})--(\ref{lastpiece})
complete the proof.
\end{proof}

\subsection{Planting a point with a small average distance to other points}
\label{agoodpointsubsection}

This subsection constructs a distance function
with respect to
which a certain point has an average distance of approximately $1/2$
to other points.

\begin{lemma}\label{boundonthenumberofmarkededges}
$$
\left|E_H^{(q(n))}\right|\le h\cdot q(n).
$$
\end{lemma}
\begin{proof}
Consider the $i$th iteration of the loop of {\sf Adv}, where
$i\in[q(n)]$.
\begin{itemize}
\item
Running
lines~4--5
results
in
$P_i$
having
at most $h$ edges.
Consequently,
\begin{eqnarray}
\left|E_H^{(i)}\right|\le \left|E_H^{(i-1)}\right|+h
\label{theincreaseofthenumberofmarkededges}
\end{eqnarray}
by line~6.
\item Running line~11
yields
$|E_H^{(i)}|=|E_H^{(i-1)}|$, implying
inequality~(\ref{theincreaseofthenumberofmarkededges})
as well.
\end{itemize}
Now,
$$
\left|
E_H^{(q(n))}
\right|
-\left|
E_H^{(0)}
\right|
=\sum_{i=1}^{q(n)}\,
\left(
\left|
E_H^{(i)}
\right|
-\left|
E_H^{(i-1)}
\right|
\right)
\stackrel{\text{(\ref{theincreaseofthenumberofmarkededges})}}{\le}
h\cdot q(n).
$$
Finally,
$|E_H^{(0)}|=0$ by
equation~(\ref{initiallymarkededgeset}).
\end{proof}

\begin{lemma}\label{fewverticeshavemanymarkedincidentedges}
$$
\left|
\left\{u\in [n]\mid \text{\rm deg}_{H^{(q(n))}}(u)\ge \delta n^{1/(h-1)}-2
\right\}
\right|
=
\frac{h}{\delta}\cdot
o(n).
$$
\footnote{We explicitly write down the
constants $h$
and $\delta$ on the right-hand side for clarity,
although they can be absorbed within $o(\cdot)$.}
\end{lemma}
\begin{proof}
By
Fact~\ref{basicgraphtheoryfact},
the average degree in $H^{(q(n))}$
is
$$
\frac{1}{n}\cdot2\cdot\left|E_H^{(q(n))}\right|.
$$
So by
the averaging argument
(that any finite nonempty sequence of nonnegative numbers
with average $\bar{a}$ has at most an $\bar{a}/t$
fraction of numbers that are greater than or equal
to $t>0$),
$$
\frac{1}{n}\cdot
\left|
\left\{u\in [n]\mid \text{\rm deg}_{H^{(q(n))}}(u)\ge \delta n^{1/(h-1)}-2
\right\}
\right|
\le
\frac{1}{n}\cdot2\cdot\left|E_H^{(q(n))}\right|\cdot\frac{1}{\delta
n^{1/(h-1)}-2},
$$
where the rightmost denominator is positive
and is $\Theta(\delta n^{1/(h-1)})$
by equation~(\ref{tediouscondition2}).
This and Lemma~\ref{boundonthenumberofmarkededges}
complete the proof.
\end{proof}

By inequality~(\ref{tediouscondition2}),
$S\setminus \{z\}\neq\emptyset$.
Let
\begin{eqnarray}
\hat{\alpha}
\stackrel{\text{def.}}{=}
\mathop{\mathrm{argmin}}_{\alpha\in S\setminus\{z\}}\,
\text{deg}_{Q^{(q(n))}}(\alpha),
\label{theoptimalpoint}
\end{eqnarray}
breaking ties arbitrarily.

\begin{lemma}\label{theoptimalpointisinvolvedinafewqueries}
For all $i\in[q(n)]$,
\begin{eqnarray}
\comment{
\left|\left\{u\in [n]\mid \left(u,\hat{\alpha}\right)\in
\left\{\left(a_i,b_i\right)
\mid i\in
\left[q(n)\right]
\right\}\right\}
\right|
}
\text{\rm deg}_{Q^{(i)}}\left(\hat{\alpha}\right)
\le \delta n^{1/(h-1)}.
\nonumber
\end{eqnarray}
\end{lemma}
\begin{proof}
By line~16 of {\sf Adv},
\begin{eqnarray}
\text{deg}_{Q^{(i)}}\left(\hat{\alpha}\right)
\le\text{deg}_{Q^{(q(n))}}\left(\hat{\alpha}\right).
\label{trivialbecausethequerygraphgrows}
\end{eqnarray}
By
equation~(\ref{theoptimalpoint})
and the averaging argument,
\begin{eqnarray}
\text{deg}_{Q^{(q(n))}}(\hat{\alpha})
\le
\frac{1}{|S\setminus\{z\}|}\cdot
\sum_{\alpha\in S\setminus\{z\}}\,
\text{deg}_{Q^{(q(n))}}(\alpha).
\nonumber
\end{eqnarray}
Furthermore,
\begin{eqnarray}
\sum_{\alpha\in S\setminus\{z\}}\,
\text{deg}_{Q^{(q(n))}}(\alpha)
\le
\sum_{\alpha\in [n]}\,
\text{deg}_{Q^{(q(n))}}(\alpha)
=2q(n),
\,\,\,
\label{sumofdegreesinthequerygraph}
\end{eqnarray}
where the
equality follows from Fact~\ref{basicgraphtheoryfact},
line~16 of {\sf Adv} and the non-repeating of queries.
Finally,
\begin{eqnarray*}
\text{deg}_{Q^{(i)}}(\hat{\alpha})
\stackrel{\text{(\ref{trivialbecausethequerygraphgrows})--(\ref{sumofdegreesinthequerygraph})}}{\le}
\frac{2q(n)}{|S|-1}
\stackrel{\text{(\ref{tediouscondition3})}}{\le}
\delta n^{1/(h-1)}.
\end{eqnarray*}
\end{proof}

Inductively,
let
\begin{eqnarray}
V_0&\stackrel{\text{def.}}{=}&\left\{\hat{\alpha}\right\},\label{layer0}\\
V_1&\stackrel{\text{def.}}{=}&
N_{Q^{(q(n))}}\left(\hat{\alpha}\right)
\setminus V_0,
\label{layer1}\\
V_{j+1}&\stackrel{\text{def.}}{=}&
N_{H^{(q(n))}}\left(V_j\right)
\setminus \left(\bigcup_{i=0}^j\, V_i\right)\label{subsequentlayers}
\end{eqnarray}
for
all
$j\in[h-2]$.
Furthermore,
\begin{eqnarray}
V_h\stackrel{\text{def.}}{=}
[n]\setminus \left(\bigcup_{i=0}^{h-1}\, V_i\right).
\label{lastlayer}
\end{eqnarray}

The following lemma is not hard to see from
equations~(\ref{layer0})--(\ref{lastlayer}).

\begin{lemma}\label{disjointnessoflayers}
$(V_0, V_1,\ldots,V_h)$ is a partition of $[n]$, i.e.,
$\bigcup_{k=0}^h\, V_k=[n]$ and $V_i\cap V_j=\emptyset$
for all distinct $i$, $j\in\{0,1,\ldots,h\}$.
\end{lemma}

Let
\begin{eqnarray}
B
&=&\left\{u\in [n]\mid \text{\rm deg}_{H^{(q(n))}}(u)\ge \delta n^{1/(h-1)}-2
\right\},\label{badset}\\
{\cal E}
&\stackrel{\text{def.}}{=}&
\left[
E_G^{(q(n))}\setminus
\left(\bigcup_{i,j\in\{0,1,\ldots,h\},\,|i-j|\ge 2}\,
V_i\times V_j\right)
\right]
\cup
\left(\left\{\hat{\alpha}\right\}\times
\left(V_h\setminus \left(B\cup S\right)
\right)\right).
\,\,\,\,\,\,\,\,\,\,\,\label{finalgraphedgeset}
\end{eqnarray}
By equation~(\ref{theoptimalpoint}),
$\hat{\alpha}\notin V_h\setminus(B\cup S)$,
which together with
equation~(\ref{completegraphedgeset}) and Lemma~\ref{chainofgraphs}
forbids any edge
in $\cal E$
from being
a self-loop.
For
all distinct $u$, $v\in [n]$,
\begin{eqnarray}
w\left(u,v\right)
\stackrel{\text{def.}}{=}
\left\{
\begin{array}{ll}
1/2,&
\text{if one of $u$ and $v$ is $\hat{\alpha}$ and the other is in $V_h\setminus (B\cup S)$,}\\
1, &\text{otherwise.}
\end{array}
\right.
\label{newedgeweightfunction}
\end{eqnarray}
Furthermore,
let
\begin{eqnarray}
{\cal G}&\stackrel{\text{def.}}{=}&\left([n],{\cal E},w\right)
\label{finalgraph}
\end{eqnarray}
be a weighted undirected graph.

\begin{lemma}\label{sizeofthenonlastlayers}
$$
\sum_{j=1}^{h-1}\,\left|V_j\right|
\le 2\delta^{h-1}n.
$$
\end{lemma}
\begin{proof}
By Lemma~\ref{boundonnumberofincidentedges} and
equation~(\ref{subsequentlayers}),
$$\left|V_{j+1}\right|
\le \left|V_j\right|\cdot \delta n^{1/(h-1)}$$
for all $j\in[h-2]$.
Therefore,
$\sum_{j=1}^{h-1}\,|V_j|$
is bounded from above by the $(h-1)$-term geometric series
with the common ratio
of $\delta n^{1/(h-1)}$
and the initial value of $|V_1|$.
Consequently,
\begin{eqnarray}
\sum_{j=1}^{h-1}\,\left|V_j\right|
\stackrel{\text{(\ref{tediouscondition2})~and~Fact~\ref{geometricseriesbound}}}{\le}
2\cdot \left|V_1\right|\cdot \delta^{h-2} n^{(h-2)/(h-1)}.
\label{sumoflayersizes}
\end{eqnarray}

By Lemma~\ref{theoptimalpointisinvolvedinafewqueries},
$|N_{Q^{(q(n))}}(\hat{\alpha})|\le \delta n^{1/(h-1)}$.
So by equation~(\ref{layer1}), we have
$|V_1|\le \delta n^{1/(h-1)}$,
which together with
inequality~(\ref{sumoflayersizes})
completes the proof.
\end{proof}

\begin{lemma}\label{manysafepointstohavesmalldistances}
$$
\left|V_h\setminus\left(B\cup S\right)
\right|
\ge n\left(1-2\delta^{h-1}
-\frac{h}{\delta}\cdot o(1)-\delta\right).
$$
\end{lemma}
\begin{proof}
By Lemma~\ref{fewverticeshavemanymarkedincidentedges}
and equation~(\ref{badset}),
$|B|=(h/\delta)\cdot o(n)$.
By construction, $|S|=\lfloor\delta n\rfloor$.
Finally,
\begin{eqnarray}
\left|V_h\right|
\stackrel{\text{Lemmas~\ref{disjointnessoflayers}--\ref{sizeofthenonlastlayers}}}{\ge}
n-2\delta^{h-1}n-\left|V_0\right|
\stackrel{\text{(\ref{layer0})}}{=}n-2\delta^{h-1}n-1.
\nonumber
\end{eqnarray}
\end{proof}

The following lemma says that $\hat{\alpha}$
has an average distance of
approximately $1/2$ to other points
w.r.t.\
the distance function
$\min\{d_{\cal G}(\cdot,\cdot),h\}$.

\begin{lemma}\label{closenesscentralityoftheoptimalpoint}
\begin{eqnarray}
\sum_{v\in [n]}\,
\min\left\{d_{\cal G}\left(\hat{\alpha},v\right),
h
\right\}
\le
n\cdot\left(\frac{1}{2}+2h\delta^{h-1}
+\frac{h^2}{\delta}\cdot o(1)+h\delta\right).
\nonumber
\end{eqnarray}
\end{lemma}
\begin{proof}
By equations~(\ref{finalgraphedgeset})--(\ref{finalgraph}),
$d_{\cal G}(\hat{\alpha},v)\le1/2$
for all $v\in V_h\setminus(B\cup S)$.
This and Lemma~\ref{manysafepointstohavesmalldistances}
complete the proof.
\end{proof}

\subsection{A metric consistent with {\sf Adv}'s answers}
\label{consistencywithametricsubsection}

This subsection constructs a metric $d\colon[n]^2\to[0,\infty)$
consistent with {\sf Adv}'s answers in line~17.
So Lemma~\ref{theclosenesscentralityoftheoutputofthealgorithm}
will require $z$, which is the output of
$A^{\sf Adv}$,
to have
an average distance (w.r.t.\ $d$) of
at least
approximately $h$ to other points.
Although $d(\cdot,\cdot)$ will not be exactly $\min\{d_{\cal G}(\cdot,\cdot),h\}$,
Lemma~\ref{closenesscentralityoftheoptimalpoint}
will
forbid
$\sum_{v\in[n]}\,d(\hat{\alpha},v)/n$ from exceeding $1/2$ by too much.
Details follow.

Recall that
$H^{(i)}$ and $G^{(i)}$
are
unweighted for all $i\in\{0,1,\ldots,q(n)\}$.
They
can be treated as
having the weight function $w$
while preserving
$d_{H^{(i)}}(\cdot,\cdot)$
and $d_{G^{(i)}}(\cdot,\cdot)$, as shown by the lemma below.

\begin{lemma}\label{consistencyofweights}
For all $i\in \{0,1,\ldots,q(n)\}$,
each
path $P$ in $H^{(i)}$ or $G^{(i)}$
has exactly $w(P)$ edges.
\end{lemma}
\begin{proof}
As $\hat{\alpha}\in S$ by equation~(\ref{theoptimalpoint}),
equation~(\ref{newedgeweightfunction}) implies
$w(u,v)=1$ for all distinct $u$, $v\in [n]\setminus S$.
This and equation~(\ref{completegraphedgeset})
imply that all edges in $E_G^{(0)}$
have weight $1$ w.r.t.\ $w$.
So by Lemma~\ref{chainofgraphs}, the edges
in $E_H^{(i)}\cup E_G^{(i)}$ have weight $1$
w.r.t.\ $w$.
Finally, recall that
$H^{(i)}=([n],E_H^{(i)})$ and $G^{(i)}=([n],E_G^{(i)})$.
\end{proof}


We now show that $H^{(q(n))}$ has an
edge in $V_i\times V_j$ only if $|i-j|\le1$.

\begin{lemma}\label{nocrossingmarkededges}
$$
E_H^{(q(n))}\cap
\left(
\bigcup_{i,j\in\{0,1,\ldots,h\},\,|i-j|\ge 2}\,
V_i\times V_j\right)=\emptyset.
$$
\end{lemma}
\begin{proof}
Suppose for contradiction that there exists
$e\in E_{H}^{(q(n))}$ with an endpoint in $V_k$ and the other in $V_\ell$,
where $k$, $\ell\in \{0,1,\ldots,h\}$ and $\ell\ge k+2$.
Then
$N_{H^{(q(n))}}(V_k)\cap V_\ell\neq\emptyset$, which
together with Lemma~\ref{disjointnessoflayers}
and $\ell\ge k+2$
implies
\begin{eqnarray}
N_{H^{(q(n))}}\left(V_k\right)\not\subseteq \bigcup_{j=0}^{k+1}\,V_j.
\label{setofneighborsdoesnotjumpacrosslayers}
\end{eqnarray}
%
As $\ell\ge k+2$ and $k$,
$\ell\in\{0,1,\ldots,h\}$,
we have
$0\le k\le h-2$.
\begin{enumerate}[{Case}~1:]
\item\label{crossingfromtheinitiallayer} $k=0$. By
equations~(\ref{theoptimalpoint})~and~(\ref{layer0}),
$V_0\subseteq S$.
So
$N_{G^{(0)}}(V_0)=\emptyset$
by
equations~(\ref{completegraphedgeset})--(\ref{initiallargegraph}).
Consequently, $N_{H^{(q(n))}}(V_0)=\emptyset$ by
Lemma~\ref{chainofgraphs},
contradicting
relation~(\ref{setofneighborsdoesnotjumpacrosslayers}).
\item\label{crossingfromamiddlelayer} $k\in[h-2]$.
Relation~(\ref{setofneighborsdoesnotjumpacrosslayers})
contradicts
equation~(\ref{subsequentlayers}) (with $j\leftarrow k$).
\end{enumerate}
A contradiction occurs in either case.
\end{proof}

\begin{lemma}\label{markededgesarepreservedinthefinalgraph}
$E_H^{(q(n))}\subseteq {\cal E}$.
\end{lemma}
\begin{proof}
By
Lemma~\ref{nocrossingmarkededges} and
equation~(\ref{finalgraphedgeset}),
$E_G^{(q(n))}\cap E_H^{(q(n))}
\subseteq{\cal E}$.
This and Lemma~\ref{chainofgraphs}
complete the proof.
\end{proof}

\begin{lemma}\label{needtogothrougheachlayer}
Let $P$ be a path in $\cal G$ that visits no edges in
$\{\hat{\alpha}\}\times (V_h\setminus(B\cup S))$.
If the first and the last vertices of $P$ are in $V_h$ and $V_1$, respectively,
then
$w(P)\ge h-1$.
\end{lemma}
\begin{proof}
By Lemma~\ref{disjointnessoflayers},
$\bigcup_{k=0}^h\, V_k=[n]$,
$V_{i+1}\cap V_i=\emptyset$ and
$(V_{i+1}\times V_i)\cap (V_{j+1}\times V_j)=
\emptyset$ for all distinct $i$, $j\in[h-1]$.
Because $P$ is a path in $\cal G$
visiting no edges in $\{\hat{\alpha}\}\times(V_h\setminus(B\cup S))$,
no edges on $P$ are in $V_i\times V_j$ for any $i$, $j\in \{0,1,\ldots,h\}$
with $|i-j|\ge 2$
by
equations~(\ref{finalgraphedgeset})~and~(\ref{finalgraph}).
This
forces
$P$, which is a
$V_h$-$V_1$ path,
to
visit at least one edge in
$V_{i+1}\times V_{i}$ for each $i\in[h-1]$
(for a total of at least $h-1$ edges).
As
$\hat{\alpha}\notin\bigcup_{i=1}^h\, V_i$
by
equations~(\ref{layer0})--(\ref{lastlayer}),
equation~(\ref{newedgeweightfunction}) gives
$w(u,v)=1$ for all
$(u,v)\in \bigcup_{i=1}^{h-1}\,V_{i+1}\times V_{i}$.
We have shown that $P$ has at least $h-1$ edges of weight (w.r.t.\ $w$)
$1$.
\end{proof}

We proceed to analyze shortest $a_i$-$b_i$ paths in ${\cal G}$, where
$i\in[q(n)]$.
Clearly,
such paths must be simple.

\begin{lemma}\label{theoptimalpointhasalongdistancetothepointswhosedistancetotheoptimalhasbeenqueriedfor}
Let $P$ be a shortest $a_i$-$b_i$ path in ${\cal G}$,
where $i\in[q(n)]$.
If $P$ visits exactly one edge in
$\{\hat{\alpha}\}\times (V_h\setminus(B\cup S))$
and $\hat{\alpha}\in \{a_i,b_i\}$,
then
$w(P)\ge h-1/2$.
\end{lemma}
\begin{proof}
Being shortest,
$P$ must be simple.
Assume $\hat{\alpha}=a_i$ for now.
Because
$P$ is a simple $\hat{\alpha}$-$b_i$ path
in ${\cal G}$
visiting exactly one edge
in $\{\hat{\alpha}\}\times (V_h\setminus(B\cup S))$,
it can be decomposed into an edge
$(\hat{\alpha},v)$, where $v\in V_h\setminus(B\cup S)$,
and
a $v$-$b_i$ path $\tilde{P}$ in ${\cal G}$
that visits no edges
in
$\{\hat{\alpha}\}\times (V_h\setminus(B\cup S))$.\footnote{If the first edge
on $P$ is not in $\{\hat{\alpha}\}\times
(V_h\setminus(B\cup S))$, then $P$'s later visit
of an edge in $\{\hat{\alpha}\}\times
(V_h\setminus(B\cup S))$ must make $P$ non-simple,
a contradiction.}
As $\hat{\alpha}=a_i$,
we have
$b_i\in N_{Q^{(q(n))}}(\hat{\alpha})$
by line~16 of {\sf Adv}.
So by
equations~(\ref{layer0})--(\ref{layer1}),
$b_i\in V_1\cup\{\hat{\alpha}\}$,
implying
$b_i\in V_1$ because querying for the distance from a point
to itself is forbidden and $\hat{\alpha}=a_i$.
In summary, $\tilde{P}$ is a
path in $\cal G$, from
$v\in V_h\setminus(B\cup S)$ to $b_i\in V_1$,
that visits
no edges in
$\{\hat{\alpha}\}\times(V_h\setminus(B\cup S))$.
So by Lemma~\ref{needtogothrougheachlayer}
(with $P\leftarrow \tilde{P}$),
\begin{eqnarray}
w\left(\tilde{P}\right)\ge h-1.
\label{thisequationnumberisaddedforbetterwriting}
\end{eqnarray}

As $v\in V_h$, we have $\hat{\alpha}\neq v$
by equations~(\ref{layer0})~and~(\ref{lastlayer}).
By the construction of $\tilde{P}$,
\begin{eqnarray}
w(P)=w\left(\hat{\alpha},v\right)+w\left(\tilde{P}\right)
\stackrel{\text{(\ref{newedgeweightfunction})}}{\ge}
\frac{1}{2}+w\left(\tilde{P}\right).
\label{thisequationnumberisaddedforbetterwriting2}
\end{eqnarray}
Inequalities~(\ref{thisequationnumberisaddedforbetterwriting})--(\ref{thisequationnumberisaddedforbetterwriting2})
show that $w(P)\ge h-1/2$.
The case of $\hat{\alpha}=b_i$ is symmetric: Reverse $P$
and exchange all the above
occurrences of ``$a_i$'' with ``$b_i$.''
\end{proof}

\begin{lemma}\label{theadjustmenttermforquerieddistancesinvolvingtheoptimalpoint}
For all $i\in[q(n)]$ with
$\hat{\alpha}\in\{a_i,b_i\}$,
$$
\chi\left[\exists v\in \left\{a_i,b_i\right\},\,
\left(v\in S\right)\land \left(\text{\rm deg}_{Q^{(i)}}(v)\le
\delta n^{1/(h-1)}
\right)
\right]
=1.
$$
\end{lemma}
\begin{proof}
By equation~(\ref{theoptimalpoint}), $\hat{\alpha}\in S$.
This and
Lemma~\ref{theoptimalpointisinvolvedinafewqueries}
complete the proof.
\end{proof}

\begin{lemma}\label{edgesarepreservedbetweengoodvertices}
For all distinct $u$, $v\in [n]\setminus(B\cup S)$,
we have
$(u,v)\in E_G^{(q(n))}$.
\end{lemma}
\begin{proof}
As $u$, $v\in [n]\setminus B$,
equation~(\ref{badset}) implies
\begin{eqnarray}
\text{\rm deg}_{H^{(i)}}(u)&<&\delta n^{1/(h-1)}-2,
\label{goodverticeshavefewmarkedincidentedges1}\\
\text{\rm deg}_{H^{(i)}}(v)&<&\delta n^{1/(h-1)}-2
\label{goodverticeshavefewmarkedincidentedges2}
\end{eqnarray}
when $i=q(n)$.
So by
Lemma~\ref{chainofgraphs},
inequalities~(\ref{goodverticeshavefewmarkedincidentedges1})--(\ref{goodverticeshavefewmarkedincidentedges2})
hold for all $i\in[q(n)]$.

As $u$, $v\in [n]\setminus S$ and $u\neq v$,
we have
$(u,v)\in E_G^{(0)}$
by equation~(\ref{completegraphedgeset}).
By lines~8~and~13 of {\sf Adv},
{\small 
\begin{eqnarray}
E_G^{(i-1)}
\setminus
\left\{\left(x,y\right)\in[n]^2
\mid
\left(\text{deg}_{H^{(i)}}(x)\ge \delta n^{1/(h-1)}-2\right)
\lor
\left(\text{deg}_{H^{(i)}}(y)\ge \delta n^{1/(h-1)}-2\right)
\right\}
\subseteq E_G^{(i)}
\label{onlytheedgesincidenttobadverticescanbedeleted}
\end{eqnarray}
}
for all $i\in[q(n)]$.
By
inequalities~(\ref{goodverticeshavefewmarkedincidentedges1})--(\ref{goodverticeshavefewmarkedincidentedges2})
and relation~(\ref{onlytheedgesincidenttobadverticescanbedeleted}),
$(u,v)\in E_G^{(i)}$ if $(u,v)\in E_G^{(i-1)}$,
for all $i\in[q(n)]$.
The proof is complete by mathematical induction.
\end{proof}

\begin{lemma}\label{twohalfedgescanbereplacedbyafulledge}
Let $P$ be a shortest $a_i$-$b_i$ path in ${\cal G}$, where
$i\in[q(n)]$.
If $P$ visits exactly two edges in
$\{\hat{\alpha}\}\times(V_h\setminus(B\cup S))$,
then
$G^{(q(n))}$
has
an
$a_i$-$b_i$ path
with
exactly $w(P)$ edges.
\end{lemma}
\begin{proof}
Being shortest,
$P$ must be simple.
Therefore, the two edges of $P$ in
$\{\hat{\alpha}\}\times(V_h\setminus(B\cup S))$,
denoted $(u,\hat{\alpha})$ and $(\hat{\alpha},v)$,
are
consecutive on $P$.
Clearly,
$u\neq v$.
Replace the
subpath
$(u,\hat{\alpha},v)$
of $P$ by
the edge
$(u,v)$ to yield an $a_i$-$b_i$ path $\tilde{P}$.
Except for the two edges of $P$ in $\{\tilde{\alpha}\}\times (V_h\setminus(B\cup S))$
(which are $(u,\hat{\alpha})$ and $(\hat{\alpha},v)$),
all edges of $P$ are in $E_G^{(q(n))}$
by equation~(\ref{finalgraphedgeset}) and $P$'s being
a path in ${\cal G}=([n],{\cal E},w)$.
As $u$, $v\in V_h\setminus(B\cup S)$ and $u\neq v$,
$(u,v)\in E_G^{(q(n))}$
by Lemma~\ref{edgesarepreservedbetweengoodvertices}.
In summary,
all the edges of $\tilde{P}$
(including $(u,v)$ and the edges of $P$ not in
$\{\hat{\alpha}\}\times(V_h\setminus(B\cup S))$)
are in $E_G^{(q(n))}$.
Consequently,
$\tilde{P}$ is
an $a_i$-$b_i$
path in $G^{(q(n))}=([n],E_G^{(q(n))})$.
%
So
we are left
only to prove that $\tilde{P}$ has exactly $w(P)$ edges, which,
by Lemma~\ref{consistencyofweights} (with $P\leftarrow \tilde{P}$
and $i\leftarrow q(n)$),
is equivalent to proving
$w(\tilde{P})=w(P)$.

Note that $\hat{\alpha}\notin V_h\setminus(B\cup S)$ by
equation~(\ref{theoptimalpoint}).
By the construction of $\tilde{P}$ and
recalling that
$u$, $v\in V_h\setminus(B\cup S)$
and $u\neq v$,
$$
w\left(\tilde{P}\right)=w(P)-w\left(u,\hat{\alpha}\right)
-w\left(\hat{\alpha},v\right)+w\left(u,v\right)
\stackrel{\text{(\ref{newedgeweightfunction})}}{=}w(P)-\frac{1}{2}-\frac{1}{2}
+1=w(P).
$$
\end{proof}

\begin{lemma}\label{nonexistenttypeofpaths}
Every simple
path in $\cal G$ visiting exactly one edge
in $\{\hat{\alpha}\}\times(V_h\setminus(B\cup S))$
either starts or ends at $\hat{\alpha}$.
\end{lemma}
\begin{proof}
By equation~(\ref{theoptimalpoint}), $\hat{\alpha}\in S$.
So by equation~(\ref{completegraphedgeset}) and Lemma~\ref{chainofgraphs},
$\hat{\alpha}$
is incident to no edges in $E_G^{(q(n))}$.
Consequently,
the set of all
edges of $\cal G$ incident to $\hat{\alpha}$
is
$\{\hat{\alpha}\}\times(V_h\setminus(B\cup S))$
by equation~(\ref{finalgraphedgeset}).
The lemma is now easy to see.
\end{proof}

\begin{lemma}\label{shortestpathsinfinalgraphdonotneededgesincidenttotheoptimalpoint}
For all $i\in[q(n)]$,
\begin{eqnarray}
&&\min\left\{d_{H^{(i)}}\left(a_i,b_i\right),
h-\frac{1}{2}
\cdot \chi\left[\exists v\in\left\{a_i,b_i\right\},\,
\left(v\in S\right)
\land\left(\text{\rm deg}_{Q^{(i)}}(v)
\le\delta n^{1/(h-1)}\right)
\right]
\right\}\nonumber\\
&\le&
\min\left\{d_{\cal G}\left(a_i,b_i\right),
h-\frac{1}{2}
\cdot \chi\left[\exists v\in \left\{a_i,b_i\right\},\,
\left(v\in S\right)\land\left(\text{\rm deg}_{Q^{(i)}}(v)
\leq\delta n^{1/(h-1)}
\right)
\right]
\right\}.\,\,\,\,\,\,\,\,\,\,\,\,\,\label{theconsistencyequation}
\end{eqnarray}
\end{lemma}
\begin{proof}
Assume the existence of an $a_i$-$b_i$
path in ${\cal G}$ for, otherwise, $d_{\cal G}(a_i,b_i)=\infty$
and inequality~(\ref{theconsistencyequation}) trivially holds.
Pick any shortest $a_i$-$b_i$ path $P$ in ${\cal G}=([n],{\cal E},w)$.
Clearly,
\begin{eqnarray}
w(P)=d_{\cal G}\left(a_i,b_i\right).\label{optimalpathweight}
\end{eqnarray}
Being shortest,
$P$ must be simple.

We establish inequality~(\ref{theconsistencyequation})
in the following exhaustive cases:
\begin{enumerate}[{Case}~1:]
\item\label{caseofnotvisitingtheoptimalpoint}
$P$ visits no edges in $\{\hat{\alpha}\}\times(V_h\setminus(B\cup S))$.
By equation~(\ref{finalgraphedgeset}),
all edges of $P$ are in $E_G^{(q(n))}$, i.e., $P$
is a path in $G^{(q(n))}$.
So by Lemma~\ref{consistencyofweights} (with $i\leftarrow
q(n)$),
$w(P)$ equals the length of $P$ in the unweighted graph $G^{(q(n))}$.
Therefore,
\begin{eqnarray}
d_{G^{(q(n))}}\left(a_i,b_i\right)
\le w(P).\label{distanceinunweightedgraphmatchesthatinthefinalgraph}
\end{eqnarray}
If $d_{G^{(i-1)}}(a_i,b_i)\le h$, then
$$
d_{H^{(i)}}\left(a_i,b_i\right)=d_{G^{(q(n))}}\left(a_i,b_i\right)
$$
by Lemma~\ref{consistentanswers}.
Otherwise, $d_{G^{(q(n))}}(a_i,b_i)>h$
by Lemma~\ref{ifthesmallestdistancerunsoutthenothersdoso}.
In either case,
equations~(\ref{optimalpathweight})--(\ref{distanceinunweightedgraphmatchesthatinthefinalgraph})
imply inequality~(\ref{theconsistencyequation}).
\item $P$ visits exactly one edge in $\{\hat{\alpha}\}\times(V_h\setminus(B\cup S))$
and $\hat{\alpha}\in\{a_i,b_i\}$.
By
Lemma~\ref{theoptimalpointhasalongdistancetothepointswhosedistancetotheoptimalhasbeenqueriedfor}
and equation~(\ref{optimalpathweight}),
$d_{\cal G}(a_i,b_i)\ge h-1/2$.
This and Lemma~\ref{theadjustmenttermforquerieddistancesinvolvingtheoptimalpoint}
force
the right-hand side of inequality~(\ref{theconsistencyequation})
to
equal
$h-1/2$.
By Lemma~\ref{theadjustmenttermforquerieddistancesinvolvingtheoptimalpoint},
the left-hand side of inequality~(\ref{theconsistencyequation}) is
less than or equal to $h-1/2$.
We have verified inequality~(\ref{theconsistencyequation}).
\item $P$ visits
exactly one edge in $\{\hat{\alpha}\}\times(V_h\setminus(B\cup S))$
and $\hat{\alpha}\notin\{a_i,b_i\}$.
A contradiction to Lemma~\ref{nonexistenttypeofpaths} occurs.
\item
$P$ visits
exactly two edges in $\{\hat{\alpha}\}\times(V_h\setminus(B\cup S))$.
Lemma~\ref{twohalfedgescanbereplacedbyafulledge}
and
that $G^{(q(n))}$ is unweighted
imply
inequality~(\ref{distanceinunweightedgraphmatchesthatinthefinalgraph}).
Proceeding as in Case~\ref{caseofnotvisitingtheoptimalpoint},
equations~(\ref{optimalpathweight})--(\ref{distanceinunweightedgraphmatchesthatinthefinalgraph})
and
Lemmas~\ref{consistentanswers}--\ref{ifthesmallestdistancerunsoutthenothersdoso}
imply inequality~(\ref{theconsistencyequation})
no matter $d_{G^{(i-1)}}(a_i,b_i)\le h$ or otherwise.
\item
$P$ visits
at least three edges in $\{\hat{\alpha}\}\times(V_h\setminus(B\cup S))$.
Clearly, $P$ is non-simple, a
contradiction.
\end{enumerate}
\end{proof}

Define $d\colon [n]^2\to[0,\infty)$ by
\begin{eqnarray}
&&d\left(a_i,b_i\right)
=d\left(b_i,a_i\right)\nonumber\\
&\stackrel{\text{def.}}{=}&
\min\left\{d_{\cal G}\left(a_i,b_i\right),
h-\frac{1}{2}
\cdot \chi\left[\exists v\in \left\{a_i,b_i\right\},\,
\left(v\in S\right)\land\left(\text{deg}_{Q^{(i)}}(v)
\leq\delta n^{1/(h-1)}
\right)
\right]
\right\},
\,\,\,\,\,\,\,\,\,\,
\label{thefinalmetriconqueriedpairs}\\
&&d\left(u,v\right)\nonumber\\
&\stackrel{\text{def.}}{=}&
\min\left\{d_{\cal G}\left(u,v\right),h\right\}
\label{thefinalmetriconnonqueriedpairs}
\end{eqnarray}
for all $i\in[q(n)]$ and
$(u,v)\in [n]^2\setminus\{(a_j,b_j)\mid j\in[q(n)]\}$.
Because all pairs in $[n]^2$ are unordered in this section,
$(b_i,a_i)\notin [n]^2\setminus\{(a_j,b_j)\mid j\in[q(n)]\}$
for all $i\in[q(n)]$.
Consequently,
equation~(\ref{thefinalmetriconnonqueriedpairs})
does not redefine $d(b_i,a_i)$.
Because ${\cal G}$ is undirected, the right-hand side
of
equation~(\ref{thefinalmetriconnonqueriedpairs})
remains intact with $u$ and $v$ interchanged.
As $A$ does not repeat queries,
equation~(\ref{thefinalmetriconqueriedpairs})
defines $d(a_i,b_i)$ and $d(b_i,a_i)$ only once for each
$i\in[q(n)]$
(note
that
forbidding repeated queries
implies
the nonexistence of distinct $i$, $j\in [q(n)]$
satisfying
(1)~$a_i=a_j$ and $b_i=b_j$ or
(2)~$a_i=b_j$ and $b_i=a_j$).
It is now clear that $d(\cdot,\cdot)$
is a well-defined function on
$[n]^2$, a set of {\em unordered} pairs.\footnote{Even if we
considered
each pair in $[n]^2$ to be ordered,
our
arguments
would
still
have shown
that
$d(\cdot,\cdot)$ is well-defined and symmetric.}
So we have the following lemma.



\begin{lemma}\label{finalmetricissymmetric}
For all $x$, $y\in [n]$,
$d(x,y)=d(y,x)$.
\end{lemma}

\comment{ 
\begin{proof}
Because the pairs in $[n]^2$
are unordered in this section,
$d(x,y)=d(y,x)$
as long as $d(\cdot,\cdot)$
is well-defined, which we have verified.
As
$\mathop{\mathrm{Im}}(w)\subseteq(0,\infty)$
by equation~(\ref{newedgeweightfunction}),
we have
$\mathop{\mathrm{Im}}(d_{\cal G})\subseteq[0,\infty]$,
which together with
equations~(\ref{thefinalmetriconqueriedpairs})--(\ref{thefinalmetriconnonqueriedpairs})
and
$h\in\mathbb{Z}^+\setminus\{1\}$
proves
$d(x,y)\ge0$.
\end{proof}
}

\begin{lemma}\label{finalmetricsatisfiesidentityofindiscernibles}
For all distinct $x$, $y\in [n]$,
$d(x,x)=0$
and $d(x,y)\ge 1/2$.
\end{lemma}
\begin{proof}
Recall that ${\cal G}=([n],{\cal E},w)$.
As $\mathop{\mathrm{Im}}(w)\subseteq[1/2,\infty)$
by equation~(\ref{newedgeweightfunction}),
we have
$d_{\cal G}(x,y)$, $d_{\cal G}(y,x)\ge 1/2$.
So by
equations~(\ref{thefinalmetriconqueriedpairs})--(\ref{thefinalmetriconnonqueriedpairs})
and
$h\in\mathbb{Z}^+\setminus\{1\}$,
$d(x,y)\ge1/2$.
Because we forbid queries for the distance from a point to itself,
$d(x,x)$ is not defined by equations~(\ref{thefinalmetriconqueriedpairs}).
By
equation~(\ref{thefinalmetriconnonqueriedpairs}),
$d(x,x)=0$.
\end{proof}

\begin{lemma}\label{lemmasayingthatwefinallyhaveametricspace}
$([n],d)$ is a metric space.
\end{lemma}
\begin{proof}
By
Lemmas~\ref{finalmetricissymmetric}--\ref{finalmetricsatisfiesidentityofindiscernibles},
we only need to
show that
\begin{eqnarray}
d\left(x,y\right)+d\left(y,z\right)\ge d\left(x,z\right)
\label{triangleinequalityforthefinalmetric}
\end{eqnarray}
for all $x$, $y$, $z\in [n]$.
It is well-known that a positively-weighted undirected graph
induces a distance function obeying the triangle
inequality; hence
\begin{eqnarray}
d_{\cal G}\left(x,y\right)
+d_{\cal G}\left(y,z\right)\ge d_{\cal G}\left(x,z\right).
\label{theknowninequalityongraph}
\end{eqnarray}

Because $\cal G$ is undirected, $d_{\cal G}(\cdot,\cdot)$ is symmetric.
So by
equations~(\ref{thefinalmetriconqueriedpairs})--(\ref{thefinalmetriconnonqueriedpairs}),
\begin{eqnarray}
d\left(x,y\right)
\in \left\{
\min\left\{d_{\cal G}\left(x,y\right),h\right\},
\min\left\{d_{\cal G}\left(x,y\right),h-\frac{1}{2}\right\}
\right\}
\label{asimplerepresentationofthefinalmetric}
\end{eqnarray}
for all $x$, $y\in [n]$.
Now
verify
inequality~(\ref{triangleinequalityforthefinalmetric})
in the following exhaustive (but not mutually exclusive) cases:
\begin{enumerate}[{Case}~1:]
\item $x=y$, $y=z$ or $x=z$.
Lemma~\ref{finalmetricsatisfiesidentityofindiscernibles}
implies inequality~(\ref{triangleinequalityforthefinalmetric}).
\item\label{casethatthefirstdistanceistruncated}
$d_{\cal G}(x,y)\ge h-1/2$ and $y\neq z$.
By relation~(\ref{asimplerepresentationofthefinalmetric}),
$d(x,y)\ge h-1/2$.
As $y\neq z$,
$d(y,z)\ge 1/2$ by
Lemma~\ref{finalmetricsatisfiesidentityofindiscernibles}.
By relation~(\ref{asimplerepresentationofthefinalmetric}),
$d(x,z)\le h$.
Summarizing the above
proves
inequality~(\ref{triangleinequalityforthefinalmetric}).
\item $d_{\cal G}(y,z)\ge h-1/2$
and $x\neq y$.
Replace ``$(x,y)$,'' ``$(y,z)$'' and ``$y\neq z$''
in the analysis of
Case~\ref{casethatthefirstdistanceistruncated}
by ``$(y,z)$,'' ``$(x,y)$'' and ``$x\neq y$,'' respectively.
\item $d_{\cal G}(x,y)<h-1/2$ and
$d_{\cal G}(y,z)<h-1/2$.
By relation~(\ref{asimplerepresentationofthefinalmetric}),
$d(x,y)=d_{\cal G}(x,y)$ and
$d(y,z)=d_{\cal G}(y,z)$.
So
inequalities~(\ref{triangleinequalityforthefinalmetric})--(\ref{theknowninequalityongraph})
share a common left-hand side.
To deduce
inequality~(\ref{triangleinequalityforthefinalmetric})
from
inequality~(\ref{theknowninequalityongraph}),
therefore,
it
suffices
to show
that
$d_{\cal G}(x,z)\ge d(x,z)$,
which follows from
relation~(\ref{asimplerepresentationofthefinalmetric}).
\end{enumerate}
\end{proof}

\begin{lemma}\label{markededgescannotmakeshorterpathsthaninthefinalgraph}
For all $i\in [q(n)]$,
$$d_{H^{(i)}}\left(a_i,b_i\right)\ge d_{\cal G}\left(a_i,b_i\right).$$
\end{lemma}
\begin{proof}
Assume the existence of an
$a_i$-$b_i$ path in $H^{(i)}$ for, otherwise,
$d_{H^{(i)}}(a_i,b_i)=\infty$
and there is nothing to prove.
Take a shortest $a_i$-$b_i$ path $P$ in the unweighted
graph $H^{(i)}=([n],E_H^{(i)})$.
So
$d_{H^{(i)}}(a_i,b_i)$ is the number of $P$'s edges.
By
Lemma~\ref{consistencyofweights},
$P$'s
number of
edges
equals
$w(P)$.
By Lemma~\ref{chainofgraphs}, $P$'s edges are in $E_H^{(q(n))}$.
So by
Lemma~\ref{markededgesarepreservedinthefinalgraph},
$P$ is a path in ${\cal G}=([n],{\cal E},w)$,
implying
$d_{\cal G}(a_i,b_i)\le w(P)$.
Summarizing the above proves the lemma.
\end{proof}

The following lemma says that line~17 of {\sf Adv}
answers queries consistently with
$d(\cdot,\cdot)$.

\begin{lemma}\label{finalconsistencylemma}
For all $i\in[q(n)]$,
\begin{eqnarray}
&&\min\left\{d_{H^{(i)}}\left(a_i,b_i\right),
h-\frac{1}{2}
\cdot \chi\left[\exists v\in\left\{a_i,b_i\right\},\,
\left(v\in S\right)
\land\left(\text{\rm deg}_{Q^{(i)}}(v)
\le\delta n^{1/(h-1)}\right)
\right]
\right\}\nonumber\\
&=&d\left(a_i,b_i\right).\label{finalconsistencyequation}
\end{eqnarray}
\end{lemma}
\begin{proof}
Lemma~\ref{shortestpathsinfinalgraphdonotneededgesincidenttotheoptimalpoint}
and equation~(\ref{thefinalmetriconqueriedpairs})
prove the ``$\le$'' part of
equation~(\ref{finalconsistencyequation}).
On the other hand,
Lemma~\ref{markededgescannotmakeshorterpathsthaninthefinalgraph}
and
equation~(\ref{thefinalmetriconqueriedpairs})
imply
the ``$\ge$'' part of
equation~(\ref{finalconsistencyequation}).
\end{proof}

\subsection{Putting things together}
\label{puttingthingstogethersubsection}

We now arrive at our main result.

\begin{theorem}\label{ourlowerbound}
{\sc Metric $1$-median}
has no deterministic $o(n^{1+1/(h-1)})$-query
$(2h-\epsilon)$-approximation
algorithms for any constants $h\in\mathbb{Z}^+\setminus\{1\}$
and $\epsilon>0$.
\end{theorem}
\begin{proof}
By
Lemma~\ref{finalconsistencylemma}
and line~17 of {\sf Adv},
{\sf Adv}
answers $A$'s queries consistently with
$d(\cdot,\cdot)$.
This
implies
that
$A^{\text{\sf Adv}}$
and $A^d$ have the same
output.\footnote{See, e.g.,~{\cite[Lemma~8]{Cha12}}.}
That is, $A^d$ outputs $z$.
By
Lemma~\ref{lemmasayingthatwefinallyhaveametricspace},
$([n],d)$ is a metric space.

By relation~(\ref{asimplerepresentationofthefinalmetric}),
$d(x,y)\le \min\{d_{\cal G}(x,y),h\}$ for all $x$, $y\in [n]$.
Therefore,
\begin{eqnarray}
\sum_{v\in [n]}\, d\left(\hat{\alpha},v\right)
\le
n\cdot\left(\frac{1}{2}+2h\delta^{h-1}
+\frac{h^2}{\delta}\cdot o(1)+h\delta
\right)
\label{qualityoftheoptimalpoint}
\end{eqnarray}
by
Lemma~\ref{closenesscentralityoftheoptimalpoint}.

Recall
that
$A$
does not repeat
queries.
So by
equation~(\ref{indexset}) and
Lemmas~\ref{finalmetricissymmetric}--\ref{finalmetricsatisfiesidentityofindiscernibles},
\begin{eqnarray}
\sum_{v\in [n]}\, d\left(z,v\right)
\ge\sum_{i\in I}\, d\left(a_i,b_i\right).
\nonumber
\end{eqnarray}
\footnote{In fact, this is an equality because
$A^{\sf Adv}$
will have queried
for the distances between its output and all other points when
halting.}
By
Lemmas~\ref{theclosenesscentralityoftheoutputofthealgorithm}~and~\ref{finalconsistencylemma},
\begin{eqnarray}
\sum_{i\in I}\, d\left(a_i,b_i\right)
\ge
n\cdot\left(h-2h\delta^{h-1}-o(1)-\delta\right).
\label{qualityoftheoutputofthealgorithmsecondpart}
\end{eqnarray}
By
inequalities~(\ref{qualityoftheoptimalpoint})--(\ref{qualityoftheoutputofthealgorithmsecondpart}),
\begin{eqnarray}
\frac{\sum_{v\in [n]}\, d\left(z,v\right)}{\sum_{v\in [n]}\,d\left(\hat{\alpha},v\right)}
\ge
\frac{h-2h\delta^{h-1}-o(1)-\delta}{1/2+2h\delta^{h-1}
+(h^2/\delta)\cdot o(1)+h\delta}.
\label{theapproximationratio}
\end{eqnarray}

Note that
all the derivations so far have been valid for all constants
$h\in\mathbb{Z}^+\setminus\{1\}$ and $\delta\in(0,1)$.
Take
$\delta=\delta(h,\epsilon)>0$ to be sufficiently small
and
$n$ to be sufficiently large
so that the right-hand side of
inequality~(\ref{theapproximationratio})
is greater than $2h-\epsilon$.\footnote{Alternatively, we may take
$$\delta=\delta(n)=\left(\frac{\max\{q(n),n\}}{n^{1+1/(h-1)}}\right)^{1/3}$$
from the beginning of this section.
Then, as
$q(n)=o(n^{1+1/(h-1)})$, the right-hand side of
inequality~(\ref{theapproximationratio}) is $2h-o(1)$, and
inequalities~(\ref{tediouscondition1})--(\ref{tediouscondition3}) remain
true for all sufficiently large $n$.}
Then
inequality~(\ref{theapproximationratio})
forbids
$z$, which is the common
output of $A^{\sf Adv}$ and $A^d$,
from being
a
$(2h-\epsilon)$-approximate $1$-median of $([n],d)$.
Note that $A$
can be any
deterministic
$o(n^{1+1/(h-1)})$-query algorithm from the beginning
of this section.
\end{proof}

Next, we use
Theorem~\ref{ourlowerbound} and
Fact~\ref{existingupperbound}
to
determine
the minimum value of $c\ge1$ such that {\sc metric $1$-median}
has a deterministic $O(n^{1+\epsilon})$-query
(resp., $O(n^{1+\epsilon})$-time) $c$-approximation
algorithm,
for each
constant
$\epsilon\in(0,1)$.


\begin{theorem}\label{maintheorem}
For
each constant
$\epsilon\in(0,1)$,
{\small 
\begin{eqnarray*}
&&
\min\left\{c\ge1
\mid
\text{{\sc metric $1$-median} has a deterministic
$O(n^{1+\epsilon})$-query
$c$-approx.\
alg.}
\right\}\\
&=&
\min\left\{c\ge1
\mid
\text{{\sc metric $1$-median} has a deterministic
$O(n^{1+\epsilon})$-time
$c$-approx.\
alg.}
\right\}\\
&=&2\left\lceil\frac{1}{\epsilon}\right\rceil.
\end{eqnarray*}
}
\comment{ 
{\footnotesize
$$
\inf\left\{c\ge1
\mid
\text{{\sc metric $1$-median} has a deterministic
$O(n^{1+\epsilon})$-query $c$-approximation
algorithm}
\right\}=2\left\lceil\frac{1}{\epsilon}\right\rceil.
$$
}
}
\end{theorem}
\begin{proof}
Take $h=\lceil1/\epsilon\rceil$; hence
$h\in\mathbb{Z}^+\setminus\{1\}$.
It is easy to verify that
$n^{1+\epsilon}=o(n^{1+1/(h-1)})$.
So by
Theorem~\ref{ourlowerbound},
{\sc metric $1$-median}
does not have a deterministic
$O(n^{1+\epsilon})$-query
$(2\lceil1/\epsilon\rceil-\epsilon')$-approximation algorithm
for any constant
$\epsilon'>0$.

Clearly,
$n^{1+1/h}=O(n^{1+\epsilon})$.
So by Fact~\ref{existingupperbound},
{\sc metric $1$-median}
has a deterministic
$O(n^{1+\epsilon})$-time
$(2\lceil1/\epsilon\rceil)$-approximation algorithm.

The above analyses remain valid with ``query'' and ``time''
exchanged because
every $O(n^{1+\epsilon})$-time algorithm
makes $O(n^{1+\epsilon})$ queries.
Consequently,
deterministic $O(n^{1+\epsilon})$-query (resp., $O(n^{1+\epsilon})$-time)
algorithms can be $(2\lceil1/\epsilon\rceil)$-approximate
but not $(2\lceil 1/\epsilon\rceil-\epsilon')$-approximate
for any constant $\epsilon'>0$.
%
\end{proof}

The brute-force
exact
algorithm
for {\sc metric $1$-median}
is well-known to
run in $O(n^2)$ time.
Therefore, there is no need to
extend
Theorem~\ref{maintheorem}
to the case of $\epsilon\ge1$.
On the other hand,
the following corollary deals with
the case of $\epsilon=0$.

\begin{corollary}
{\sc Metric $1$-median} does not
have a deterministic $O(n^{1+o(1)})$-query
(resp., $O(n^{1+o(1)})$-time)
$O(1)$-approximation algorithm.
\end{corollary}
\begin{proof}
Take
$h\to\infty$
in
Theorem~\ref{ourlowerbound}.
\end{proof}

\section*{Acknowledgments}

The author is supported in part by the Ministry of Science and Technology of
Taiwan under grant 103-2221-E-155-026-MY2.

\appendix
\section{Optimizing the hidden factors in Theorem~\ref{ourlowerbound}}

This appendix
discusses
how
the
bound of $o(n^{1+1/(h-1)})$
in Theorem~\ref{ourlowerbound}
hides
factors
dependent on $h$.
For all $i\in[q(n)]$,
\begin{eqnarray}
B_{i-1}
\stackrel{\text{def.}}{=}
\left\{v\in[n]\mid \text{deg}_{H^{(i-1)}}(v)
\ge \delta n^{1/(h-1)}-2
\right\}.
\label{badpointsinamiddlelevel}
\end{eqnarray}

\begin{lemma}\label{goodverticesdonothavetheiredgesremoved}
For all $i\in[q(n)]$ and distinct $u$, $v\in [n]\setminus(B_{i-1}\cup S)$,
we have
$(u,v)\in E_G^{(i-1)}$.
\end{lemma}
\begin{proof}
As $u$, $v\in[n]\setminus B_{i-1}$,
\begin{eqnarray*}
\text{deg}_{H^{(j)}}(u)&<&\delta n^{1/(h-1)}-2,\\
\text{deg}_{H^{(j)}}(v)&<&\delta n^{1/(h-1)}-2
\end{eqnarray*}
for all $j\in\{0,1,\ldots,i-1\}$
by equation~(\ref{badpointsinamiddlelevel})
and Lemma~\ref{chainofgraphs}.
So by lines~8~and~13 of {\sf Adv},
$(u,v)\in E_G^{(j)}$ if $(u,v)\in E_G^{(j-1)}$,
for all $j\in[i-1]$.
By equation~(\ref{completegraphedgeset}),
$(u,v)\in E_G^{(0)}$.
The proof is complete by mathematical induction.
\end{proof}

\begin{lemma}\label{shortestpathiscomposedmainlyofbadpoints}
For each $i\in[q(n)]$ such that the $i$th iteration
of the loop of {\sf Adv} runs lines~5--9,
$P_i$ in line~5 does not have two
non-consecutive
vertices in $[n]\setminus(B_{i-1}\cup S)$.
\end{lemma}
\begin{proof}
By
line~5 of {\sf Adv},
two
non-consecutive
vertices on $P_i$
are not connected
by an edge in $E_G^{(i-1)}$.
This and
Lemma~\ref{goodverticesdonothavetheiredgesremoved}
complete the proof.
\end{proof}

\begin{lemma}\label{forabadvertextheincidentedgesareremoved}
For all $i\in[q(n)]$ and $v\in B_{i-1}$,
$$N_{G^{(i-1)}}(v)\subseteq N_{H^{(i-1)}}(v).$$
\end{lemma}
\begin{proof}
By equation~(\ref{badpointsinamiddlelevel}),
$$
\text{deg}_{H^{(i-1)}}(v)\ge\delta n^{1/(h-1)}-2.
$$
Clearly,
$$\text{\rm deg}_{H^{(0)}}(v)
\stackrel{\text{(\ref{initiallymarkededgeset})}}{=}0
\stackrel{\text{(\ref{tediouscondition2})}}{<}\delta n^{1/(h-1)}-2.$$
So there exists $j\in[i-1]$ satisfying
\begin{eqnarray}
\text{deg}_{H^{(j-1)}}(v)&<&\delta n^{1/(h-1)}-2,
\label{thenumberofmarkedincidentedgesbeforethejump}\\
\text{deg}_{H^{(j)}}(v)&\ge&\delta n^{1/(h-1)}-2.\label{thejustjumpednumberofmarkededges}
\end{eqnarray}

Clearly,
\begin{eqnarray}
N_{G^{(j)}}(v)
=\left\{u\in[n]\mid (u,v)\in E_G^{(j)}
\right\}.
\label{straightforwarddefinitionofneighborset}
\end{eqnarray}
As $H^{(j-1)}\neq H^{(j)}$ by
inequalities~(\ref{thenumberofmarkedincidentedgesbeforethejump})--(\ref{thejustjumpednumberofmarkededges}),
the $j$th iteration of the loop of {\sf Adv}
runs lines~5--9 but not 11--14.
By inequality~(\ref{thejustjumpednumberofmarkededges})
and line~8 of {\sf Adv},
\begin{eqnarray}
\left\{u\in[n]\mid (u,v)\in E_G^{(j)}
\right\}
=\left\{u\in[n]\mid (u,v)\in E_G^{(j-1)}\setminus
\left(E_G^{(j-1)}\setminus E_H^{(j)}\right)
\right\}.
\label{thesetofedgesthatremaininGj}
\end{eqnarray}
Equations~(\ref{straightforwarddefinitionofneighborset})--(\ref{thesetofedgesthatremaininGj})
and Lemma~\ref{chainofgraphs}
give
$$N_{G^{(j)}}(v)=N_{H^{(j)}}(v).$$
This and Lemma~\ref{chainofgraphs} complete the proof.
\end{proof}

\begin{lemma}\label{markingatmostoneedgeeachround}
For all $i\in[q(n)]$,
\begin{eqnarray}
\left|E_H^{(i)}\right|
\le
\left|E_H^{(i-1)}\right|+1.
\nonumber
\end{eqnarray}
\end{lemma}
\begin{proof}
Clearly,
we may assume that the $i$th iteration of the loop of {\sf Adv}
runs lines~5--9 but not 11--14.
By line~6,
we only need
to show that
\begin{eqnarray}
\left|\left\{e\mid\left(\text{$e$ is an edge on $P_i$}\right)
\land \left(e\notin E_H^{(i-1)}\right)\right\}\right|\le1.
\label{thetransformedgoalforprovingthatonlyoneedgeisnewlymarked}
\end{eqnarray}
By Lemma~\ref{shortestpathiscomposedmainlyofbadpoints},
$P_i$ in line~5 has at most one edge in $([n]\setminus(B_{i-1}\cup S))^2$.
So, to prove
inequality~(\ref{thetransformedgoalforprovingthatonlyoneedgeisnewlymarked}),
it suffices to show that
each edge $(u,v)$
on $P_i$ with
$(u,v)\notin ([n]\setminus(B_{i-1}\cup S))^2$
satisfies $(u,v)\in E_H^{(i-1)}$,
as done below:
\begin{enumerate}[{Case}~1:]
\item $\{u,v\}\cap S\neq\emptyset$.
By equation~(\ref{completegraphedgeset})
and Lemma~\ref{chainofgraphs},
$(u,v)\notin E_G^{(i-1)}$.
Consequently, $P_i$ has an edge not in $E_G^{(i-1)}$,
contradicting line~5.
\item $\{u,v\}\cap B_{i-1}\neq\emptyset$.
By symmetry, assume $v\in B_{i-1}$.
So by Lemma~\ref{forabadvertextheincidentedgesareremoved},
$N_{G^{(i-1)}}(v)\subseteq N_{H^{(i-1)}}(v)$.
Because $P_i$ is a path in $G^{(i-1)}$ by line~5 and
$(u,v)$ is on $P_i$,
$u
\in N_{G^{(i-1)}}(v)$.
In summary,
$u\in N_{H^{(i-1)}}(v)$.
I.e.,
$(u,v)\in E_H^{(i-1)}$.
\end{enumerate}
\end{proof}

The following improvement over
Lemma~\ref{boundonthenumberofmarkededges}
is immediate from
equation~(\ref{initiallymarkededgeset})
and Lemma~\ref{markingatmostoneedgeeachround}.

\begin{lemma}\label{strongboundonthenumberofmarkededges}
$$
\left|E_H^{(q(n))}\right|
\le q(n).$$
\end{lemma}

Assuming
$100\le h=o(n^{1/(h-1)})$,
the following modifications
to this paper show that
the bound of $o(n^{1+1/(h-1)})$
in Theorem~\ref{ourlowerbound}
depends
on $h$
as $o(n^{1+1/(h-1)}/h)$:
\begin{enumerate}[(1)]
\item\label{basicsettingfortheparameterswithlargeh} Take
\begin{eqnarray*}
q(n)&=&o\left(\frac{n^{1+1/(h-1)}}{h}\right),\\
\delta
&=&
h\cdot\frac{\max\{q(n),n\}}
{n^{1+1/(h-1)}},\\
\lambda&=&\delta^{h/8},\\
S&=&[\lfloor\lambda n
\rfloor].
\end{eqnarray*}
\item Replace ``$\delta$''
by ``$\sqrt{\delta}$''
in
inequality~(\ref{tediouscondition2}).
\item Replace ``$\delta$''
by $1/\delta^{h/4}$
in
inequality~(\ref{tediouscondition3}).
\item Replace
the two occurrences of
``$\delta$''
by ``$\sqrt{\delta}$''
in line~8 of {\sf Adv}.
\item Replace
``$\delta$''
by ``$1/\delta^{h/4}$''
in line~17 of {\sf Adv}.
\item Replace all
occurrences of ``$\delta$''
by ``$\sqrt{\delta}$''
in Lemma~\ref{boundonnumberofincidentedges}
and its proof.
\item Replace all
occurrences of ``$\delta$''
by ``$\sqrt{\delta}$''
in Lemma~\ref{averagedistanceishighinthemarkedgraph}
and its proof.
\item Replace
``$\delta n^{1/(h-1)}$''
and
``$h-2h\delta^{h-1}-o(1)-\delta$''
by ``$n^{1/(h-1)}/\delta^{h/4}$''
and
``$h
-2h{\sqrt{\delta}}^{h-1}
-o(1)
-\lambda/2
-1/(2\delta^{h/4}n^{1-1/(h-1)})$,''
respectively,
in
the statement of
Lemma~\ref{theclosenesscentralityoftheoutputofthealgorithm}.
\item Replace all
occurrences of
``$\delta n^{1/(h-1)}$,''
``$2\delta^{h-1} n$''
and ``$\lfloor\delta n\rfloor$''
by ``$n^{1/(h-1)}/\delta^{h/4}$,''
``$2{\sqrt{\delta}}^{h-1} n$''
and ``$\lfloor\lambda n\rfloor$,''
respectively,
in the proof of
Lemma~\ref{theclosenesscentralityoftheoutputofthealgorithm}.
\item Replace all
occurrences of
``$\delta n^{1/(h-1)}$,''
``$(h/\delta)\cdot o(n)$''
and
``Lemma~\ref{boundonthenumberofmarkededges}''
by ``$\sqrt{\delta}\,
n^{1/(h-1)}$,''
``$(1/\sqrt{\delta})\cdot O(q(n)/n^{1/(h-1)})$''
and
``Lemma~\ref{strongboundonthenumberofmarkededges},''
respectively,
in
Lemma~\ref{fewverticeshavemanymarkedincidentedges}
and its proof.
\item That
$\hat{\alpha}$ is well-defined in
equation~(\ref{theoptimalpoint})
follows from $|S|\ge2$, which
holds
for all sufficiently large $n$
by
item~(\ref{basicsettingfortheparameterswithlargeh})
and $h\ge100$.
\item Replace all occurrences
of ``$\delta$''
by ``$1/\delta^{h/4}$''
in
Lemma~\ref{theoptimalpointisinvolvedinafewqueries}
and its proof.
\item Replace
``$\delta$''
by ``$\sqrt{\delta}$''
in
equation~(\ref{badset}).
\item Replace ``$\delta^{h-1}$''
by ``$\delta^{h/4-1}$''
in the statement of
Lemma~\ref{sizeofthenonlastlayers}.
\item Replace
all occurrences of ``$\delta$''
by ``$\sqrt{\delta}$''
and ``$1/\delta^{h/4}$,''
respectively,
in the first and the second
paragraphs of the proof of
Lemma~\ref{sizeofthenonlastlayers}.
\item Replace
``$1-2\delta^{h-1}
-(h/\delta)\cdot o(1)-\delta$''
by
``$1-2\delta^{h/4-1}
-(1/\sqrt{\delta})\cdot O(q(n)/n^{1+1/(h-1)})
-\lambda
$''
in the statement of
Lemma~\ref{manysafepointstohavesmalldistances}.
\item Replace all occurrences
of ``$(h/\delta)\cdot o(n)$,''
``$\lfloor\delta n\rfloor$''
and ``$\delta^{h-1}$''
by ``$(1/\sqrt{\delta})\cdot O(q(n)/n^{1/(h-1)})$,''
``$\lfloor\lambda n\rfloor$''
and ``$\delta^{h/4-1}$,''
respectively,
in the proof of
Lemma~\ref{manysafepointstohavesmalldistances}.
\item Replace ``$\delta^{h-1}$,''
``$(h^2/\delta)\cdot o(1)$'' and ``$h\delta$''
by ``$\delta^{h/4-1}$,''
``$(h/\sqrt{\delta})\cdot O(q(n)/n^{1+1/(h-1)})$''
and ``$h\lambda$,'' respectively,
in the statement of
Lemma~\ref{closenesscentralityoftheoptimalpoint}.
\item Replace ``$\delta$''
by ``$1/\delta^{h/4}$''
in the statement of
Lemma~\ref{theadjustmenttermforquerieddistancesinvolvingtheoptimalpoint}.
\item Replace all occurrences of ``$\delta$''
by ``$\sqrt{\delta}$''
in the proof of
Lemma~\ref{edgesarepreservedbetweengoodvertices}.
\item Replace the two
occurrences of ``$\delta$''
by ``$1/\delta^{h/4}$''
in the statement of
Lemma~\ref{shortestpathsinfinalgraphdonotneededgesincidenttotheoptimalpoint}.
\item Replace ``$\delta$''
by ``$1/\delta^{h/4}$''
in
equation~(\ref{thefinalmetriconqueriedpairs}).
\item Replace ``$\delta$''
by ``$1/\delta^{h/4}$''
in the statement of
Lemma~\ref{finalconsistencylemma}.
\item Replace ``$\delta^{h-1}$,''
``$(h^2/\delta)\cdot o(1)$'' and ``$h\delta$''
by ``$\delta^{h/4-1}$,''
``$(h/\sqrt{\delta})\cdot O(q(n)/n^{1+1/(h-1)})$''
and ``$h\lambda$,'' respectively,
in
inequality~(\ref{qualityoftheoptimalpoint}).
\item Replace
``$h-2h\delta^{h-1}-o(1)-\delta$''
by
``$h-2h{\sqrt{\delta}}^{h-1}-o(1)-\lambda/2-1/(2\delta^{h/4}n^{1-1/(h-1)})$''
in the right-hand side of
inequality~(\ref{qualityoftheoutputofthealgorithmsecondpart}).
\item Replace
the numerator and the denominator
on the right-hand side of
inequality~(\ref{theapproximationratio})
by
``$h-2h{\sqrt{\delta}}^{h-1}
-o(1)-\lambda/2-1/(2\delta^{h/4}n^{1-1/(h-1)})$''
and
``$1/2+2h\delta^{h/4-1}
+(h/\sqrt{\delta})\cdot O(q(n)/n^{1+1/(h-1)})
+h\lambda$,''
respectively.
\item
Verify
that the
right-hand side of
inequality~(\ref{theapproximationratio})
is $2h-o(1)$.
To see this,
use
item~(\ref{basicsettingfortheparameterswithlargeh})
and $100\le h=o(n^{1/(h-1)})$
to
verify that
$\delta=o(1)$,
$\max_{x\ge 1}\, x\cdot \delta^{x/8}=O(\delta)=o(1)$ (which requires
elementary calculus and reveals that
$h\sqrt{\delta}^{h-1}=o(1)$, $h\delta^{h/4-1}=o(1)$
and $h\lambda=h\delta^{h/8}=o(1)$),
$\lambda=o(1)$,
$\delta^{h/4}\ge 1/n^{h/(4(h-1))}$,
$\delta^{h/4}\cdot n^{1-1/(h-1)}=n^{\Omega(1)}$,
$\sqrt{\delta}\ge \sqrt{h\cdot q(n)/n^{1+1/(h-1)}}$ and
$\sqrt{h\cdot q(n)/n^{1+1/(h-1)}}=o(1)$.
\item Replace all
occurrences of ``$\delta$''
by ``$\sqrt{\delta}$''
in
equation~(\ref{badpointsinamiddlelevel})
as well as
in the proofs
of
Lemmas~\ref{goodverticesdonothavetheiredgesremoved}~and~\ref{forabadvertextheincidentedgesareremoved}.
\end{enumerate}

\bibliographystyle{plain}
\bibliography{tradeoffmedian}

\noindent

\end{document}